\documentclass[11pt,a4paper]{article}

\usepackage[utf8]{inputenc}
\usepackage[T1]{fontenc}

\usepackage{amsmath,amssymb,amsthm}
\newtheorem{theorem}{Theorem}[section]

\usepackage{mathtools}

\usepackage[hidelinks]{hyperref}
\usepackage{xurl}

\usepackage[margin=2.5cm]{geometry}
\usepackage{setspace}
\usepackage{titling}
\usepackage{authblk}
\usepackage{graphicx}
\usepackage{booktabs}
\usepackage{array}
\usepackage{tabularx}
\usepackage{caption}
\usepackage{algorithm}
\usepackage{algpseudocode}
\usepackage{listings}

\usepackage[numbers,sort&compress]{natbib}

\usepackage{xcolor}
\usepackage{tikz}
\usepackage{seqsplit}
\usepackage{enumitem}
\setlist{itemsep=2pt, topsep=4pt}

\providecommand{\keywords}[1]{\vspace{0.5em}\noindent\textit{Keywords---} #1}

\title{snaproot: Decentralized File Integrity Verification \\
       Using Blockchain-Anchored Cryptographic Hashing}

\author[]{Arslan Br\"omme\thanks{Dipl.-Inform., B.Sc., CISSP, CISA, CISM, CAISE.
                            \texttt{arslanb@chain-horizon.com}}
{ }and Tarkan Yavas\thanks{M.Eng., CISM, ISO27001 Lead Auditor, CAISE.
                              \texttt{tarkan.yavas@chain-horizon.com}}}
\affil[]{Chain Horizon GmbH}

\date{Version 0.9.9.3 (Working Draft)}

\begin{document}
\maketitle

\begin{center}
\fbox{\parbox{0.92\linewidth}{\centering\small\itshape
Preprint / working paper. This version is a work in progress and may be
updated. Comments are welcome.}}
\end{center}
\vspace{0.5em}

\begin{abstract}
The rapid growth of digital content has made reliable integrity verification
increasingly important. Existing solutions rely either on centralized
authorities, which introduce trust dependencies and single points of failure,
or on decentralized storage systems that incur prohibitive resource overhead.
In this paper, we present \emph{snaproot}, a lightweight system that
implements the hash-anchoring paradigm of Haber and Stornetta on the Solana
blockchain to provide efficient, decentralized file integrity verification.
snaproot generates a SHA-256 hash of a file and stores it immutably on-chain
as a permanent reference record. Verification is performed by recomputing
the hash and comparing it to the stored value, yielding a deterministic
binary outcome. Beyond the basic anchoring mechanism, we describe a
four-tier trust architecture comprising three realized tiers --- user
sovereignty through locally generated keys, within-chain immutability
through consensus, and cross-fork verifiability --- and one prospective
tier consisting of a planned cross-chain migration mechanism for
long-term persistence beyond the lifetime of any single blockchain. We present a formal threat
model, a security analysis grounded in the second-preimage resistance of
SHA-256, and an empirical evaluation on Solana Devnet across file sizes
from 1\,KB to 500\,MB, with sub-five-second confirmation latency on Devnet that
supports capture-time anchoring of photographs as a practical
workflow for typical network conditions. The system is deployed on Solana Mainnet;
a formal characterization under Mainnet load is left to future work.
A central conceptual contribution of this work is the explicit
structural separation between \emph{existence proof} --- the
chain-resident, key-independent claim that a file existed at a given
time --- and \emph{authorship proof} --- the key-dependent binding
between a record and a specific wallet identity. This separation,
which is absent or implicit in most prior hash-anchoring designs,
allows existence guarantees to survive key loss, organisational
change, and provider disappearance, while preserving the ability to
make stronger authorship claims where keys are retained. We position
snaproot against established hash-anchoring systems (OpenTimestamps,
OriginStamp, Chainpoint) and discuss limitations with respect to
pre-registration manipulation and AI-generated content.
\end{abstract}

\keywords{blockchain, file integrity, cryptographic hashing, decentralized
verification, Solana, digital trust, capture-time anchoring}

\section{Introduction}
\label{sec:introduction}

The integrity of digital files is fundamental to a wide range of
applications, including legal documentation, digital forensics, medical
recordkeeping, and content authentication. Unlike physical documents,
digital files can be silently modified without leaving any visible trace,
making post-hoc verification difficult~\cite{menezes1996handbook}. This
challenge has grown significantly more complex with the proliferation of
generative artificial intelligence, which enables the creation of highly
realistic synthetic content that is often indistinguishable from authentic
material by visual inspection alone~\cite{rossler2019faceforensics}.

Traditional approaches to data integrity rely on centralized mechanisms
such as digital signatures issued by certificate authorities, proprietary
timestamping services, or notarial registries. While these approaches can
function adequately in controlled environments, they share a structural
limitation that no implementation quality can remove: the integrity claim
ultimately depends on the continued correct operation of a single, trusted
authority. The 2011 compromise of the DigiNotar certificate authority
illustrates this limitation concretely: a single intrusion enabled the
issuance of hundreds of fraudulent certificates, and users had no
independent mechanism to verify whether a given certificate had been
manipulated or legitimately issued~\cite{prins2011diginotar}. Centralized
authorities can be coerced, compromised, or simply disappear, taking the
integrity claims they support with them.

Blockchain technology offers a structurally different alternative. A
distributed, publicly auditable ledger removes the dependency on any
single party for record persistence. Once data is written and confirmed,
it cannot be altered without network-wide consensus, making blockchains
an attractive substrate for long-term integrity anchoring. However,
directly storing large files on-chain is prohibitively expensive and
technically impractical. A lightweight alternative is to store only the
cryptographic hash of a file, which serves as a compact and verifiable
fingerprint~\cite{nakamoto2008bitcoin}. The hash-anchoring paradigm,
originally proposed by Haber and Stornetta~\cite{haber1991timestamp},
operationalizes this principle and has been instantiated in systems such
as OpenTimestamps~\cite{todd2016ots}, OriginStamp~\cite{gipp2015originstamp},
and Chainpoint~\cite{vaughan2016chainpoint}.

In this paper, we present \emph{snaproot}, a system that combines SHA-256
hashing with on-chain hash anchoring on the Solana blockchain. snaproot is
a concrete realization of the hash-anchoring paradigm on a high-throughput
blockchain substrate, with confirmation latencies in the single-digit
seconds range rather than the minutes or hours required by Bitcoin-based
predecessors. This latency profile is a qualitative enabler for use cases that are impractical on slower
substrates under typical conditions, most notably
\emph{capture-time anchoring}, in which a photograph is hashed
and registered on-chain within the user's natural shutter-to-preview
window. The system is deployed on Solana Mainnet and is in productive use.

\subsection{The Trust Architecture}

Rather than positioning snaproot primarily as a performance improvement
over existing hash-anchoring systems, this paper develops a more
fundamental claim: snaproot realizes a layered trust architecture in
which integrity claims depend on no single trusted party and remain
verifiable across organizational, technological, and generational
boundaries. We describe this architecture in four tiers:

\begin{itemize}
\item \textbf{User sovereignty.} Each snaproot installation generates a
Solana keypair locally. The private key never leaves the user's device.
Records are signed under the user's own wallet, and verification requires
only the original file, a transaction identifier, and access to any
Solana RPC endpoint --- no interaction with snaproot infrastructure is
needed. The user is structurally independent of the provider; if snaproot
as an organization were to cease operation, existing records would remain
fully verifiable.
\item \textbf{Within-chain integrity.} Once confirmed on Solana, records
are protected by consensus and the second-preimage resistance of SHA-256.
A modification of a registered file is detectable with deterministic
certainty, and reversing a confirmed record would require breaking
consensus on a supermajority of staked validators.
\item \textbf{Cross-fork verifiability.} Past records survive
protocol-level disruptions of the substrate. The Solana network has
undergone protocol changes and recovery events; however, these affect the
liveness of the network for new anchorings, not the safety of records
already on archived chain history. As long as at least one archive of the
relevant chain segment remains accessible, the existence of an anchored
record can be verified, and any tampering with chain history becomes
detectable through comparison of independently maintained archives.
\item \textbf{Cross-chain trust migration.} snaproot is designed with
the prospect of substrate independence in mind. In the planned
mechanism, the hashes recorded on Solana could be aggregated by the
user and re-anchored on a different blockchain (for example, Ethereum)
through a user-initiated, wallet-bound migration. This mechanism is
described as an outlook rather than an active development item; its
purpose, if realized, would be to ensure that the choice of Solana
remains operational rather than architectural, so that the system can
outlive the lifetime of any single blockchain.
\end{itemize}

An important consequence of this architecture is the structural separation
of two properties that are conflated in many alternative designs. The
\emph{existence proof} for a registered file --- the cryptographic claim
that the file existed in its current form at the time of registration ---
is independent of any specific key custody. It can be verified by anyone
with access to the file and a working RPC endpoint, indefinitely, and
remains intact even if the registrant's private key is later lost. The
\emph{authorship proof} --- the cryptographic binding between the record
and a specific wallet identity --- is key-dependent and is forfeited if
the private key is lost. This separation allows snaproot to support both
identified and pseudonymous integrity claims, and to provide existence
guarantees that survive scenarios in which authorship claims would not.

\subsection{Applications}

The trust architecture is general and supports a range of applications.
We discuss two illustrative domains in this paper.

\emph{Media authenticity against AI-generated content.} The deterministic
nature of hash comparison provides a counterpoint to the inherently
probabilistic approaches based on AI-based detection. Where a detector
classifies an image as authentic or manipulated with some confidence
that degrades against unfamiliar generators~\cite{rossler2019faceforensics,
li2018ictuoculi}, snaproot answers a different and stricter question:
whether a specific file has changed since its on-chain registration. When
combined with capture-time anchoring --- registering the photograph at
the moment of capture, before any modification opportunity --- this
shifts the trust boundary from "the asserted creator" to "the capture
pipeline of a specific app and device", a substantially smaller and more
auditable surface.

\emph{Tamper-evident records for compliance.} Audit logs, regulatory
records, and other documentation subject to integrity requirements
typically rely on organizationally controlled storage or signed archives.
The integrity claim then depends on the same organization that may also
benefit from manipulation. Anchoring such records to a decentralized
ledger removes this conflict: an auditor or regulator can verify
integrity externally, without relying on the audited organization's own
infrastructure.

\subsection{Contributions}

The main contributions of this work are:

\begin{enumerate}
\item A four-tier trust architecture for blockchain-anchored integrity
verification, of which three tiers are fully realized in the deployed
system --- user sovereignty, within-chain integrity, and cross-fork
verifiability --- and one tier (substrate-independent cross-chain
migration) is developed as a coherent prospective design rather than
an implemented mechanism. We identify and discuss the structural separation
of existence proof and authorship proof, which clarifies the security
guarantees of the system and distinguishes them from those of alternative
designs.
\item A practical Solana implementation of the architecture, deployed on
Mainnet, including capture-time anchoring of photographs as a functional
demonstration of how the substrate's latency characteristics enable
specific application classes that are structurally infeasible on
slower-confirmation substrates.
\item A formal threat model with a security reduction to the
second-preimage resistance of SHA-256, covering non-registrant
modification, replay, timestamp manipulation, and duplicate registration,
together with an empirical evaluation on Solana Devnet measuring hashing
performance, confirmation latency, cost, and sustained throughput.
\end{enumerate}

The remainder of this paper is organized as follows.
Section~\ref{sec:related} reviews related work.
Section~\ref{sec:threat} defines the threat model and problem formulation.
Section~\ref{sec:architecture} describes the system architecture and
protocols, including capture-time anchoring.
Section~\ref{sec:security} provides a security analysis along the four
tiers of the trust architecture.
Section~\ref{sec:evaluation} presents the empirical evaluation.
Section~\ref{sec:limitations} discusses limitations, the active roadmap,
and the longer-term outlook including cross-chain migration.
Section~\ref{sec:discussion} discusses the trust architecture in the
context of media authenticity and compliance applications.
Section~\ref{sec:conclusion} concludes.

\section{Related Work}
\label{sec:related}

We discuss prior work in five areas relevant to snaproot's positioning:
blockchain-based timestamping and hash anchoring, decentralized storage,
on-chain ownership records and NFTs, AI-based content authentication, and
tamper-evident logging systems. Throughout, we focus on the specific
properties that distinguish snaproot from each line of work, rather than
on a comprehensive survey.

\subsection{Blockchain-Based Timestamping and Hash Anchoring}
\label{sec:related:anchoring}

The use of a tamper-evident chain to establish temporal ordering and
provenance of digital documents was pioneered by Haber and
Stornetta~\cite{haber1991timestamp}, predating modern blockchains by
nearly two decades. Several systems instantiate this concept on public
blockchains. OpenTimestamps~\cite{todd2016ots} anchors hash values on
the Bitcoin blockchain using Merkle aggregation, allowing a single
Bitcoin transaction to attest to arbitrarily many documents. Its design
is fully permissionless and provides strong long-term guarantees by
inheriting Bitcoin's settlement assurances; however, Bitcoin's
confirmation window of approximately sixty minutes makes it unsuitable
for latency-sensitive applications. OriginStamp~\cite{gipp2015originstamp},
while presented by its authors as a decentralized timestamping service,
in practice relies on a centralized aggregation backend that batches
client hashes before anchoring their Merkle root to Bitcoin; this
operational centralization partially undermines its end-to-end trust
guarantees, as users depend on the continued operation and honesty of
the aggregation service. Chainpoint~\cite{vaughan2016chainpoint} and
related hash-aggregation services such as Factom, Woleet, and Po.et
similarly provide Merkle-tree-based anchoring on Bitcoin or Ethereum,
typically operated as managed services that batch many client hashes
into a single on-chain root.

These systems share with snaproot the basic hash-anchoring paradigm but
differ along three properties central to the present work. First,
\emph{user sovereignty}: in snaproot, each installation generates a
Solana keypair locally and the user signs each record under their own
wallet, so that the verification path does not require the operator's
infrastructure. In contrast, Merkle-aggregation services typically place
the user's hashes inside an operator-controlled aggregation step, and
verification depends on the operator's continued willingness and
ability to provide Merkle proofs. Second, \emph{latency profile}:
snaproot's sub-five-second confirmation latency on Solana (Devnet; see Section~\ref{sec:evaluation}) supports use cases that are impractical on Bitcoin-based substrates under typical conditions, most notably the capture-time anchoring of photographs discussed in
Section~\ref{sec:architecture}. Third, \emph{substrate independence}:
snaproot is designed such that hashes anchored on one chain can be
aggregated and re-anchored on another at user initiative, decoupling the
long-term verifiability of records from the lifetime of any single
blockchain. We return to this point in Sections~\ref{sec:security}
and~\ref{sec:limitations}.

The current snaproot implementation anchors one hash per transaction and
does not perform on-chain Merkle aggregation; we discuss this design
choice and its implications in Section~\ref{sec:limitations}.

Table~\ref{tab:related-work} summarises the comparison along the
dimensions most relevant to the trust architecture developed in this
paper.

\begin{table}[!htbp]
\centering
\small
\caption{Comparison of hash-anchoring systems along trust-architecture
dimensions. \checkmark\ = property holds by design;
$\circ$\ = partially or with caveats;
--\ = property absent or structurally precluded.
Provider-independence: verification requires no operator cooperation.
User-signed records: each record carries the user's own cryptographic
signature. Capture-time anchoring: confirmation latency supports
synchronous photo registration. Existence/authorship separation:
system explicitly separates the two proof types.}
\label{tab:related-work}
\begin{tabularx}{\linewidth}{X c c c c}
\toprule
\textbf{Property}
  & \textbf{OTS}
  & \textbf{OriginStamp}
  & \textbf{Chainpoint}
  & \textbf{snaproot} \\
\midrule
Provider-independent verification
  & $\circ$ & -- & $\circ$ & \checkmark \\
User-signed records
  & -- & -- & -- & \checkmark \\
Per-user authorship proof
  & -- & -- & -- & \checkmark \\
Capture-time anchoring
  & -- & -- & -- & \checkmark \\
Existence / authorship separation
  & -- & -- & -- & \checkmark \\
No operator aggregation step
  & \checkmark & -- & -- & \checkmark \\
Long-term substrate independence
  & \checkmark & $\circ$ & $\circ$ & $\circ^\dagger$ \\
Permissionless access
  & \checkmark & $\circ$ & $\circ$ & \checkmark \\
\bottomrule
\multicolumn{5}{l}{\footnotesize $^\dagger$Prospective; cross-chain
migration is a planned outlook feature, not yet implemented.}
\end{tabularx}
\end{table}

The most consequential differences are in the first four rows.
Provider-independent verification and user-signed records are
architectural properties that derive from snaproot's local-wallet
model; they are absent in Merkle-aggregation services because the
user's hash is processed inside an operator-controlled aggregation
step, and verification depends on the operator supplying the Merkle
proof. The explicit separation of existence proof and authorship proof
is a conceptual contribution of the present work and is not discussed
in the prior systems. Capture-time anchoring is a latency-driven
property that is structurally precluded on Bitcoin-based substrates.
Long-term substrate independence is the one dimension on which
OpenTimestamps holds a genuine advantage: Bitcoin's archive ecosystem
is more decentralised and longer-established than Solana's, and
snaproot's cross-chain migration path is prospective rather than
realised.

\subsection{Decentralized Storage Systems}

IPFS~\cite{benet2014ipfs} provides a content-addressed distributed
storage layer in which files are identified by their content hash (CID).
While this enables integrity checking as a side effect, the system is
designed primarily for data distribution and availability, not for
authoritative integrity verification: the same CID can be served by any
participating node, and long-term persistence depends on external pinning
services or the Filecoin incentive layer~\cite{protocollabs2017filecoin},
introducing variable costs and operational dependencies. snaproot takes
a fundamentally different approach: file content is never transmitted or
stored by the system. Only the hash is anchored on-chain, eliminating
storage overhead while retaining full integrity verification capability.
The two approaches are complementary rather than competing: a workflow
that requires both content availability and authoritative integrity
could combine IPFS-based storage with snaproot-based anchoring.

\subsection{Blockchain Ownership and NFTs}

Non-fungible tokens (NFTs) associate digital assets with on-chain
ownership records~\cite{entriken2018erc721}. However, many NFT
implementations store asset data or metadata off-chain and reference it
through a TokenURI or external storage service. Prior work on ERC-721
metadata permanence shows that a substantial share of NFT assets depend
on non-permanent, centralized, or mutable
storage~\cite{barrington2022nft}, meaning the on-chain ownership record
alone does not guarantee file integrity --- an asset can be transferred
on-chain while its underlying content is silently replaced off-chain.
snaproot differs fundamentally by anchoring the hash of the actual file
contents directly on-chain. Any subsequent modification of the file is
immediately detectable, regardless of how the file is stored or
distributed. Ownership and integrity are orthogonal properties; snaproot
addresses the latter without making claims about the former.

\subsection{AI-Based Content Authentication}

A growing body of literature applies machine learning to detect
AI-generated or manipulated content~\cite{rossler2019faceforensics,
li2018ictuoculi}. Although these approaches have achieved high accuracy
on known generative architectures, they are inherently probabilistic and
degrade as generation quality improves. R\"ossler et
al.~\cite{rossler2019faceforensics} showed, in the context of
FaceForensics++, that detectors trained on a single manipulation method
generalize substantially worse than those trained across the full set of
manipulation methods in the dataset, indicating strong sensitivity to
the training distribution. Content authenticity standards such as
C2PA~\cite{c2pa2023spec} embed signed provenance metadata at creation
time, which is complementary but requires tool-level adoption by content
producers and depends on the integrity of the signing infrastructure.

snaproot is orthogonal to both approaches: rather than classifying
content, it verifies whether a previously registered file has been
altered. The verification answer is binary and deterministic, with
soundness reducing to the second-preimage resistance of SHA-256. Used
together with point-of-capture mechanisms --- either C2PA-conformant
hardware or the capture-time anchoring path implemented in snaproot
itself --- the deterministic post-hoc verification provided by anchoring
and the provenance signal provided by capture-time mechanisms reinforce
each other.

\subsection{Tamper-Evident Logging and Certificate Transparency}

Certificate Transparency~\cite{laurie2013ct} uses append-only,
cryptographically verifiable Merkle tree logs to provide auditable
records of TLS certificates. This architecture shares snaproot's core
tamper-evident property but is domain-specific and relies on a fixed
set of designated log operators; users must trust the federated operator
set to behave honestly. snaproot generalizes the concept to arbitrary
files using a fully permissionless blockchain substrate, eliminating
the need for trusted log operators. The motivating example from
Section~\ref{sec:introduction} --- the DigiNotar compromise of
2011~\cite{prins2011diginotar} --- illustrates the failure mode that
both Certificate Transparency and snaproot are designed to make
detectable: a trusted authority that quietly behaves dishonestly. The
two systems address this concern in different domains and at different
layers, but with structurally similar primitives.

\section{Threat Model and Problem Formulation}
\label{sec:threat}

We define the adversary model, security goals, and assumptions under
which snaproot operates. The goals are structured to make explicit the
four tiers of the trust architecture introduced in
Section~\ref{sec:introduction}: user sovereignty, within-chain integrity,
cross-fork verifiability, and substrate-independent migration.

\subsection{Adversary Model}
\label{sec:threat:adversary}

We consider a probabilistic polynomial-time (PPT) adversary $\mathcal{A}$
with full read access to the public state of the Solana blockchain,
including all anchored hash records, transaction metadata, and timestamps.
The adversary may also observe the broader public environment: the
existence of files registered by other users, the wallet pubkeys of
registrants, and any metadata associated with anchored records.

The primary threat that snaproot must resist is the post-registration
substitution of a registered file by a different file that nonetheless
passes verification. Given a registered file $F$ with hash
$H = \mathrm{SHA256}(F)$, this corresponds to $\mathcal{A}$ finding
$F' \neq F$ with $\mathrm{SHA256}(F') = H$ --- that is, a second-preimage
attack against SHA-256. Specifically, $\mathcal{A}$ may attempt to:

\begin{itemize}
\item Observe any registered hash value and the public key of the
registrant, and any associated metadata.
\item Modify the contents of a registered file $F$ to produce $F'$ and
present $F'$ as passing verification against the original registration.
This is a second-preimage attack and constitutes the primary threat.
\item Forge or retroactively modify a blockchain record so as to
validate a modified file.
\item Mount a full preimage attack: given $H$, attempt to find any $F'$
with $\mathrm{SHA256}(F') = H$. This is strictly harder than the
second-preimage attack and is included for completeness.
\item Mount a collision attack: find $F$ and $F'$ with $F \neq F'$ and
$\mathrm{SHA256}(F) = \mathrm{SHA256}(F')$. This threat is primarily
relevant in the stronger sub-case of a malicious registrant who commits
to $F$ but later attempts to substitute a precomputed $F'$. It therefore
affects the non-repudiation property rather than the core integrity
guarantee.
\end{itemize}

\paragraph{Out of scope.} We assume that $\mathcal{A}$ cannot break the
computational hardness assumptions underlying SHA-256 and cannot compromise
the consensus mechanism of the Solana network. Network-level attacks
against Solana's validator set are outside the scope of application-layer
integrity verification.

\paragraph{Frontrunning as a known limitation.} A specific concern
that the current implementation does not actively defend against is
\emph{timestamp frontrunning} of registration transactions. The
primary planned mitigation is the salted anchoring construction
$H_{\mathrm{anchor}} = \mathrm{SHA256}(F \,\|\, s)$ described in
Section~\ref{sec:limitations:roadmap}; we introduce the threat here
so that the mitigation in the roadmap can be read in context. An
adversary $\mathcal{A}$ who learns of an unconfirmed registration
--- or who otherwise obtains the file before the legitimate user
has anchored it --- may submit a competing memo transaction with the
same hash under their own wallet, attempting to be confirmed first
and thereby establishing an earlier on-chain existence claim. Because
snaproot does not enforce protocol-level uniqueness on hashes
(Section~\ref{sec:architecture:onchain}), the legitimate user's own
anchor is not displaced; their existence claim remains valid from
their own anchor time onward. The damage is therefore narrower than
displacement: the adversary can claim that the file existed earlier
than it actually did under their identity, which manufactures a
competing narrative about provenance even though it does not invalidate
the legitimate registrant's record. In practice, this attack requires
$\mathcal{A}$ to know the file $F$ before its registration is
confirmed, which restricts the threat to scenarios where the file is
leaked, intercepted, or shared with a malicious party prior to
anchoring. A planned salt mechanism (Section~\ref{sec:limitations})
mitigates this by computing the anchor as
$\mathrm{SHA256}(F\,\|\,s)$ with a user-controlled salt $s$, so that
an adversary in possession of $F$ alone cannot construct a competing
anchor that matches the legitimate one. We discuss this trade-off
explicitly in Section~\ref{sec:limitations}.

\subsection{Security Goals}
\label{sec:threat:goals}

Let $F$ denote a file, $H = \mathrm{SHA256}(F)$ its hash, $W$ the
public key of the registering wallet, and $B$ the blockchain ledger.
snaproot is designed to satisfy the following properties:

\begin{description}
\item[(G1) Integrity.] If $B$ contains $H$ at time $t$, then any
$F' \neq F$ submitted for verification at time $t' > t$ produces
$H' \neq H$ (except with negligible probability), and the verification
protocol returns FALSE.

\item[(G2) Existence-proof persistence.] The claim ``a file with hash
$H$ existed at time $t$'' remains independently verifiable for any party
with access to $F$ and to any archive of $B$ that includes the record,
\emph{regardless of the continued availability of the registrant's
private key}. This property is structural: it follows from the
persistence of $B$ and from the public verifiability of the hash, and
is orthogonal to authorship.

\item[(G3) Non-repudiation (authorship).] So long as the private key
corresponding to $W$ is held by the registrant, the cryptographic
binding between the record and $W$ is irrevocable: the registrant
cannot deny having signed the registration. If the private key is
lost, the authorship claim becomes unprovable, but Integrity~(G1) and
Existence-proof persistence~(G2) remain intact.

\item[(G4) Public verifiability.] Any party with read access to $B$ can
independently perform verification, given $F$ and the transaction
identifier. No cooperation or involvement of the original registrant is
required.

\item[(G5) Provider independence.] Verification of any record requires
only $F$, the transaction identifier, and access to any functioning
Solana RPC endpoint. It does not require interaction with snaproot's
infrastructure, services, or web frontend, nor with any specific
operator. Records remain verifiable even if snaproot as an organization
ceases to operate.

\item[(G6) Cross-fork verifiability.] If $B$ undergoes a fork or
protocol-level disruption at time $t_{\text{fork}}$, records anchored
before $t_{\text{fork}}$ remain verifiable on any chain version that
shares the pre-fork history, provided at least one archive of that
history is accessible. Discrepancies between competing chain versions
after $t_{\text{fork}}$ are detectable through comparison, rather than
hidden.
\end{description}

The separation between (G2) and (G3) is central to the trust
architecture and is discussed further in Section~\ref{sec:security}.
It is what distinguishes snaproot from systems in which the loss of a
single secret invalidates both the existence claim and the authorship
claim simultaneously.

\subsection{Assumptions}
\label{sec:threat:assumptions}

The security analysis in Section~\ref{sec:security} rests on the
following assumptions:

\begin{itemize}
\item \textbf{Hash function security.} SHA-256 is second-preimage
resistant: no PPT adversary, given a fixed $F$, can find $F' \neq F$
with $\mathrm{SHA256}(F') = \mathrm{SHA256}(F)$ with non-negligible
probability. Collision resistance, a strictly stronger property, is
additionally assumed for the malicious-registrant sub-case.

\item \textbf{Solana safety.} We use the terms \emph{safety} and
\emph{liveness} in their established distributed-systems
sense~\cite{lamport1977multiprocess} --- as subsequently standardized
in the distributed-systems literature --- where safety means that no
two honest participants accept conflicting histories, so that confirmed
blocks are never reverted; liveness means that the network continues
to make progress by confirming new transactions. For
records already confirmed on $B$, we assume that the Solana consensus
mechanism maintains safety: confirmed blocks are not reverted, and the
history up to a confirmed block remains consistent across honest
validators. We distinguish this from liveness, the ability of the
network to confirm new transactions in a timely manner, which we do
not assume to hold indefinitely. Past records depend on safety; future
anchorings depend on both safety and liveness.

\item \textbf{Pre-registration integrity.} The registrant performs
registration before any adversarial modification of the file. snaproot
provides no guarantee if an already-modified file is registered.
Combining snaproot with capture-time anchoring
(Section~\ref{sec:architecture}) tightens this assumption by reducing
the window between file creation and registration to within the
substrate's confirmation latency.

\item \textbf{Archive availability for long-term verification.} For
records to remain verifiable years or decades after registration, at
least one archive of the relevant chain history must remain accessible.
We do not assume that the same archive remains accessible indefinitely,
nor that any specific operator preserves it; we only assume that some
archive of the relevant pre-fork history exists. Section~\ref{sec:limitations}
discusses the cross-chain trust migration mechanism that further weakens
this assumption.
\end{itemize}

\section{System Architecture}
\label{sec:architecture}

snaproot consists of three logical components: a Client Layer, a
Hashing Module, and a Blockchain Layer. The Client Layer is realized as
a user-facing application; the Hashing Module operates locally within
that application; and the Blockchain Layer is implemented as a Solana
program (smart contract) deployed on Solana Mainnet. We describe each
component, then specify the registration and verification protocols
and the capture-time anchoring workflow.

\subsection{Overview}
\label{sec:architecture:overview}

The data flow for the two main operations is summarized in
Figure~\ref{fig:flow}. In a registration, the user submits a file to
the Client Layer; the Hashing Module computes its SHA-256 digest; the
Client Layer constructs a Solana transaction containing this hash
together with metadata, signs it with the user's locally held key, and
submits it to the Blockchain Layer for inclusion on Solana. In a
verification, the user submits a file and the transaction identifier
of a prior registration; the Hashing Module recomputes the hash; the
Client Layer retrieves the stored hash from the Blockchain Layer and
compares the two values, yielding a deterministic binary result.

\begin{figure}[ht]
\centering
\small
\begin{tabular}{l}
\textbf{Registration:}\\
\quad User $\to$ Client Layer $\to$ Hashing Module $\to$ Blockchain Layer\\[0.5em]
\textbf{Verification:}\\
\quad User $\to$ Client Layer $\to$ Hash Comparison $\leftarrow$ Stored Hash\\
\end{tabular}
\caption{System architecture and data flow for registration (top) and
verification (bottom).}
\label{fig:flow}
\end{figure}

\subsection{Client Layer and Local Key Management}
\label{sec:architecture:client}

The Client Layer provides the user-facing interface for file submission,
verification, and result presentation. It accepts files in any format,
forwards raw bytes to the Hashing Module, manages Solana wallet
interactions for transaction signing, and displays results.

A key property of the Client Layer is its handling of cryptographic
keys. During account registration, each snaproot installation generates
a fresh Solana keypair (Ed25519) locally; the wallet is created as an
integral part of the onboarding flow and is presented to the user at
sign-up time. The private key is stored exclusively on the user's mobile device
in the application's local secure storage and is never transmitted to
snaproot servers, cloud services, or any third-party service. All
transaction signing is performed locally on the device; the private
key does not leave the device boundary at any point in the
registration or verification workflow. The Blockchain Layer accepts
only transactions signed by the user's own key. This realizes property~(G5) of the threat model
(Section~\ref{sec:threat:goals}): verification of any record requires
only the file, the transaction identifier, and access to any Solana
RPC endpoint, with no dependency on snaproot infrastructure for
correctness.

The current implementation allows the user to export the private key
from the snaproot wallet; the exported key can be imported into any
standard Solana wallet provider, so that the user retains access to
the wallet identity independently of snaproot. Importing an exported
key back into another snaproot installation is a feature scheduled for
the next release; until then, a user operating snaproot on multiple
mobile devices generates a distinct wallet on each. The user-sovereignty
property of the trust architecture is unaffected either way: the private
key is held by the user, not by snaproot. We return to this point in
Section~\ref{sec:limitations}.

The Client Layer is intentionally stateless with respect to file
content: no file data is retained after a request is processed. Anchor
records and their associated metadata are retained locally for the
user's reference, but the cryptographically relevant data --- hash,
registrant pubkey, timestamp --- lives on Solana.

\subsection{Hashing Module}
\label{sec:architecture:hashing}

The Hashing Module applies SHA-256~\cite{nist2015fips180} to the raw
byte representation of the submitted file. Processing is performed
entirely locally; file contents are never transmitted to any external
service. SHA-256 was chosen for its widespread standardization, its
strong security properties under current cryptographic assumptions,
and its high throughput on commodity hardware. The empirical
characterization of hashing throughput across file sizes is reported in
Section~\ref{sec:evaluation:hashing}.

In the current implementation, the anchor hash is computed directly as
$H_{\text{anchor}} = \mathrm{SHA256}(F)$, where $F$ is the raw file
content. A planned extension introduces an optional user-controlled
salt parameter (Section~\ref{sec:limitations}), permitting independent
parties to anchor identical files without conflict and providing
additional resistance against frontrunning.

\subsection{Blockchain Layer and On-Chain Data Model}
\label{sec:architecture:onchain}

snaproot anchors each registration as a single Solana memo transaction.
There is no dedicated on-chain account per registration: the
information that constitutes a snaproot record --- the file hash and
the metadata --- is carried as the payload of a transaction that
invokes the standard Solana memo program (program identifier
\texttt{MemoSq4gqABAXKb96qnH8TysNcWxMyWCqXgDLGmfcHr}).

The choice of the memo program rather than a dedicated snaproot smart
contract reflects three deliberate design decisions. First,
\emph{no rent-exempt deposit}: dedicated Solana accounts require a
minimum SOL balance to remain rent-exempt; memo transactions require
no such deposit, so the record persists in ledger history without
any ongoing cost. Second, \emph{permissionless verifiability}: the
memo program is a standard, audited Solana program whose behaviour is
publicly known; a verifier needs no knowledge of a proprietary
snaproot contract to retrieve and interpret a record. Third,
\emph{reduced attack surface}: the integrity of the anchoring
mechanism depends only on the correctness of the standard memo
program, not on the security of any snaproot-specific on-chain logic.
The trade-off is that the memo program does not enforce any
application-level constraints on payload structure; snaproot's
\texttt{SR1} version prefix (Section~\ref{sec:architecture:payload})
provides the necessary format discipline at the application layer.

Once a memo transaction is confirmed, it becomes a permanent part of
the Solana ledger and can be retrieved by its transaction signature.

\paragraph{Payload structure and versioning.}\label{sec:architecture:payload} The memo payload is a
UTF-8 plaintext string in a pipe-separated format introduced in
payload format version~SR1. A representative payload has the
structure:
\begin{center}
\texttt{SR1|<SHA256>|<timestamp\_ms>|<gps\_or\_null>|<salt>|<source>|<filetype>}
\end{center}
where \texttt{SR1} is the format version prefix, \texttt{<SHA256>}
is the 64-character hex-encoded SHA-256 hash of the file,
\texttt{<timestamp\_ms>} is the Unix timestamp in milliseconds at
the time of client-side submission, \texttt{<gps\_or\_null>} is the
GPS coordinate string or \texttt{null} if location is unavailable or
denied, \texttt{<salt>} is reserved for the planned optional salt
(Section~\ref{sec:limitations:roadmap}), \texttt{<source>} identifies
the capture method (e.g.\ \texttt{camera}), and \texttt{<filetype>}
identifies the file type (e.g.\ \texttt{foto}). The version prefix
\texttt{SR1} allows future payload formats to be distinguished
unambiguously from the current one; a verifier encountering an
unknown prefix can signal that the record requires a newer client
version. The Solana memo program imposes a payload size limit well in
excess of what a snaproot record requires; a representative
\texttt{SR1} payload occupies approximately 112~bytes.

\paragraph{Authorship and fee payment.} A snaproot transaction is
signed by the user's locally held wallet, which appears in the
transaction as the signer of the memo and therefore as the
cryptographic author of the record (property G3,
Section~\ref{sec:threat:goals}). Transaction fees are paid by a
separate snaproot-operated wallet, the \emph{relay wallet}, which
acts as both the Solana fee payer and a co-signer of every user
transaction. The Solana Program Instruction Logs confirm this
explicitly: each confirmed snaproot memo transaction records
\texttt{Signed by [relay\_wallet]} and \texttt{Signed by
[user\_wallet]}, with the relay wallet tagged as \emph{Fee Payer},
\emph{Signer}, and \emph{Writable}, and the user wallet tagged as
\emph{Signer} only. This structure has four consequences worth noting
explicitly. First, the user does not need to hold SOL to use
snaproot; the relay wallet covers the network fee. Second, the
cryptographic authorship of each record remains with the user,
because the user's wallet --- not the relay wallet --- is the memo
signer, and this binding cannot be altered after confirmation. Third,
the relay wallet's co-signature constitutes a verifiable
\emph{infrastructure authenticity signal}: because the relay
wallet's public key is known and published by snaproot, any verifier
can confirm that a given memo transaction was submitted through the
snaproot relay infrastructure rather than by an independent party
mimicking the memo format. A party wishing to fabricate a
snaproot-like record that passes this check must control the relay
wallet's private key. Fourth, and importantly, a compromise of the
relay wallet would affect only this infrastructure authenticity
signal: it would not affect file integrity~(G1), existence-proof
persistence~(G2), or authorship~(G3), all of which depend solely on
the user's own key material and the immutability of the confirmed
transaction. We discuss this threat boundary explicitly in
Section~\ref{sec:security:relay}. The relay wallet is an operational
co-signer, not a custodian of records.

\paragraph{Implicit on-chain attributes.} Three further attributes of a
registration are derivable from the confirmed transaction itself, without
appearing as separate fields in the payload. The \emph{block time}
recorded by the Solana runtime serves as the registration timestamp; it
is set by consensus, not by the client. The \emph{signer} of the memo
identifies the registrant. The \emph{transaction signature} (tx\_id) is
the canonical handle for later retrieval of the record.

\paragraph{No protocol-level uniqueness.} Because each registration is
an independent memo transaction, there is no protocol-enforced
uniqueness constraint on hashes: the same hash may be anchored in
multiple memo transactions, potentially under different user wallets.
The first confirmed memo transaction for a given hash is the earliest
on-chain existence proof for that file; subsequent anchorings of the
same hash do not displace it, they simply add later existence proofs.
We return to the implications of this in
Section~\ref{sec:security:duplicate}.

\paragraph{Metadata storage.} The metadata, like the hash, is part of
the memo payload and is therefore stored on-chain in cleartext. It is
not part of the SHA-256 input over the file. Being on-chain, the
metadata inherits the same immutability as the rest of the
transaction once confirmed; what it does \emph{not} inherit is a
cryptographic binding to the file, since it does not enter the
hash computation. A planned extension
(Section~\ref{sec:limitations:roadmap}) will let the user choose, per
use case, whether metadata is anchored at all and, if so, in what form.
We discuss the implications of the current design choice in
Section~\ref{sec:security}.

\subsection{Registration Protocol}
\label{sec:architecture:register}

Algorithm~\ref{alg:register} specifies the registration protocol. The
client computes the hash, constructs the memo payload from the hash
and metadata, assembles a Solana transaction with the user's wallet as
the memo signer and the snaproot relay wallet as the fee payer, has it
co-signed and submitted, and waits for confirmation. There is no
on-chain existence check at registration time, since no protocol-level
uniqueness is enforced (Section~\ref{sec:architecture:onchain}).

\begin{algorithm}[ht]
\caption{\textsc{RegisterFile}($F$, \textit{wallet}, \textit{metadata})}
\label{alg:register}
\begin{algorithmic}[1]
\State $H \gets \mathrm{SHA256}(F)$
\State $\textit{payload} \gets \textsc{Encode}(H, \textit{metadata})$
\State $\textit{tx} \gets \textsc{BuildMemoTx}(\textit{payload},$
    \State \quad\quad $\textit{signer} = \textit{wallet},$
    \State \quad\quad $\textit{fee\_payer} = \textit{relay\_wallet})$
\State $\textit{tx} \gets \textsc{SignByUser}(\textit{tx}, \textit{wallet})$
\State $\textit{tx} \gets \textsc{SignByRelay}(\textit{tx}, \textit{relay\_wallet})$
\State $\textit{tx\_id} \gets \textsc{SubmitToBlockchain}(\textit{tx})$
\State $\textsc{AwaitConfirmation}(\textit{tx\_id})$
\State \Return $\textit{tx\_id}$
\end{algorithmic}
\end{algorithm}

\subsection{Verification Protocol}
\label{sec:architecture:verify}

Algorithm~\ref{alg:verify} specifies the verification protocol. The
verifier recomputes the hash from the file, retrieves the memo
transaction from Solana by its transaction signature, extracts the
anchored hash from the memo payload, and compares the two values. The
result is a deterministic binary outcome.

\begin{algorithm}[ht]
\caption{\textsc{VerifyFile}($F$, \textit{tx\_id})}
\label{alg:verify}
\begin{algorithmic}[1]
\State $H_{\text{computed}} \gets \mathrm{SHA256}(F)$
\State $\textit{tx} \gets \textsc{FetchTransaction}(\textit{tx\_id})$
\If{$\textit{tx} = \mathrm{NULL}$}
    \State \Return \textsc{Error}(``Transaction not found on chain'')
\EndIf
\State $\textit{payload} \gets \textsc{ExtractMemoPayload}(\textit{tx})$
\State $H_{\text{stored}} \gets \textsc{ParseHash}(\textit{payload})$
\If{$H_{\text{computed}} = H_{\text{stored}}$}
    \State \Return TRUE \Comment{Integrity confirmed}
\Else
    \State \Return FALSE \Comment{File has been modified}
\EndIf
\end{algorithmic}
\end{algorithm}

The verification protocol takes the transaction signature as input
because it is the canonical handle for the on-chain record. A planned
extension (Section~\ref{sec:limitations:roadmap}) will additionally
support verification when only the file and the user's wallet public
key are known: in that case, the verifier scans the memo transactions
of the user's wallet and tests each one's anchored hash against the
computed hash of the file, succeeding as soon as a match is found.

Verification can be executed against any Solana RPC endpoint: the
official Mainnet endpoints, a self-hosted node, or third-party
block explorers that expose the same RPC interface. snaproot itself
provides a web-based verification frontend as a convenience, but as
established by property~(G5), this frontend is not on the critical
verification path: the cryptographic operations --- hash recomputation
and comparison --- can be performed by any party with access to the
file and the transaction identifier.

\subsection{Capture-Time Anchoring}
\label{sec:architecture:capture}

In addition to the basic register/verify workflow described above, the
snaproot application implements \emph{capture-time anchoring} for
photographs. The workflow is integrated directly into the camera
function of the application:

\begin{enumerate}
\item The user takes a photograph through the snaproot camera interface.
\item The application presents a pre-submission screen showing a preview
of the captured image, an editable filename field, and the current GPS
availability status. If location access has been granted, the GPS
coordinates are shown; if GPS is unavailable or has been denied, the
screen indicates \emph{``GPS not available --- will be snaprooted without
coordinates''} and proceeds without geolocation metadata. The user
confirms by tapping the \emph{snaproot} button.
\item Upon confirmation, the application computes the SHA-256 hash of the
image bytes as delivered by the device operating system's camera API
to the application. This is the precise object of the integrity
guarantee: snaproot anchors \emph{the image bytes as presented to the
application}, not a claim about the physical capture process upstream
of the OS camera API.
\item The hash is immediately submitted to the Blockchain Layer for
anchoring on Solana Mainnet, signed by the installation's local wallet.
\item In parallel, the image is saved to the device's normal photo
library, with the corresponding anchor record (transaction identifier,
on-chain timestamp) retained by the application for later reference.
\end{enumerate}

The operational significance of this workflow is twofold. First, it
narrows the window between content creation and on-chain anchoring to
within the substrate's confirmation latency --- a single-digit number
of seconds on Solana --- so that any adversary wishing to substitute a
manipulated image for the original must do so within this window or
must control the capture device itself. The trust boundary shifts from
``the asserted creator'' to ``the capture pipeline of a specific
device and application instance''.

A precise characterisation of what is and is not proven is important
for avoiding overclaims. snaproot proves that a specific byte sequence
existed at registration time and has not changed since. Applied to a
photograph taken through the snaproot camera interface, it proves
that \emph{the image bytes delivered by the OS camera API to the
application} existed at that time. It does not independently prove
that those bytes originated directly from the physical camera sensor
without intermediate processing by the OS, firmware, or hardware
image signal processor (ISP). In practice, for images captured in the
snaproot camera workflow on unmodified consumer hardware, the gap
between the physical capture and the OS-delivered bytes is narrow and
auditable; but the formal guarantee is scoped to the bytes the
application received, not to the physical light that entered the
lens.

Second, the workflow demonstrates that the Solana latency profile is
a qualitatively different operating regime relative to Bitcoin-based
anchoring systems, enabling this class of application under typical
network conditions. A workflow that required the user to wait approximately
sixty minutes for a Bitcoin confirmation, or ten minutes for an
aggregated anchor to be batched and submitted, would not be acceptable
in the natural photography use case: the user would have moved on, the
shutter context would be lost, and the integration with the device's
normal photo library would break down. Confirmation in the seconds
range, by contrast, fits within the natural shutter-to-preview window
on a smartphone and allows the anchoring to remain invisible to the
user from a workflow perspective.

We discuss the security implications of capture-time anchoring,
including its relationship to the post-publication versus
pre-publication adversary distinction, in Section~\ref{sec:discussion}.

\section{Security Analysis}
\label{sec:security}

We analyze snaproot's security properties along the goals defined in
Section~\ref{sec:threat:goals}. The analysis covers the three realized
tiers of the trust architecture: user sovereignty
(Section~\ref{sec:security:provider}), within-chain integrity
(Sections~\ref{sec:security:integrity}--\ref{sec:security:nonrep}), and
cross-fork verifiability (Section~\ref{sec:security:crossfork}). The
fourth tier --- cross-chain trust migration --- is a prospective design
that is not yet implemented; its security properties are discussed
conceptually in Section~\ref{sec:limitations}.

Goal~(G4) (public verifiability) does not receive a separate subsection
because it follows directly from~(G5): if verification requires only the
file, the transaction identifier, and any functioning RPC endpoint, then
any party who possesses these inputs can independently perform
verification --- which is exactly what (G4) asserts. The argument is
subsumed in Section~\ref{sec:security:provider}.

\subsection{Integrity (G1)}
\label{sec:security:integrity}

\begin{theorem}[File Integrity]
\label{thm:integrity}
Let $F$ be a file registered under hash $H = \mathrm{SHA256}(F)$,
and let $\mathcal{A}$ be a PPT adversary. Assuming SHA-256 is
second-preimage resistant, the probability that $\mathcal{A}$ produces
$F' \neq F$ such that Algorithm~\ref{alg:verify} returns \textsc{True}
on input $(F', \mathit{tx\_id})$ is negligible in the security
parameter~$\lambda$.
\end{theorem}

\begin{proof}[Proof sketch]
Algorithm~\ref{alg:verify} returns \textsc{True} if and only if
$\mathrm{SHA256}(F') = \mathrm{SHA256}(F) = H$. For $F' \neq F$, this
requires $\mathcal{A}$ to find a second preimage of $H$ under SHA-256.
By the second-preimage resistance assumption, no PPT adversary can do
so with non-negligible probability. Therefore the probability of a
false \textsc{True} is bounded by
$\mathsf{Adv}^{\mathrm{2pre}}_{\mathrm{SHA256}}(\mathcal{A}, \lambda)$,
which is negligible.
\end{proof}

In the stronger sub-case where the registrant is adversarial, the
relevant assumption is collision resistance: a malicious registrant
may commit to $F$ and later attempt to substitute a precomputed $F'$
with $\mathrm{SHA256}(F') = \mathrm{SHA256}(F)$. Collision resistance
of SHA-256, a strictly stronger property, is also widely assumed to
hold~\cite{nist2015fips180}. The integrity guarantee therefore covers
both the honest-registrant and malicious-registrant cases under
their respective hardness assumptions.

\subsection{Existence-Proof Persistence (G2)}
\label{sec:security:existence}

A central design property of snaproot is the separation between the
existence proof for a registered file and the authorship proof. We
make this property explicit here because it differs from the security
model of many alternative systems, in which both proofs collapse onto
the same secret.

The existence proof for a registered file $F$ is the verifiable claim
that ``a file with hash $H = \mathrm{SHA256}(F)$ existed at time $t$,
with $t$ established by the on-chain block timestamp of the
registration''. We claim that this proof remains verifiable for any
party with access to $F$ and to any archive of $B$ containing the
record, \emph{regardless of the continued availability of the
registrant's private key}.

The argument is direct. The verification procedure of
Algorithm~\ref{alg:verify} requires only $F$ and the transaction
identifier; it does not require a signature operation by the
registrant. The transaction identifier and the corresponding record
are public information on the blockchain, accessible to any party
with read access to $B$. The integrity of the record (G1) ensures that
the value $H$ retrieved from $B$ matches the registered hash, and the
hash comparison $\mathrm{SHA256}(F) \stackrel{?}{=} H$ is a public
computation that does not require any secret material from the
registrant. The existence proof is therefore structural: it inherits
its persistence from the persistence of $B$, not from any party's key
custody.

The practical consequences are substantial. A verifier who acquires a
file $F$ years after its original registration can establish the
existence proof autonomously, even if the original registrant has long
since lost their private key, lost their device, or ceased to be
identifiable. snaproot thereby provides a class of integrity claims
that survive scenarios in which conventional signature-based systems
would silently fail.

\subsection{Non-Repudiation and the Authorship Proof (G3)}
\label{sec:security:nonrep}

The authorship proof is the cryptographic binding between a record and
the wallet identity $W$ that signed it. This proof is established at
registration time: the transaction is signed by the private key
corresponding to $W$, and the Solana consensus mechanism only
incorporates transactions with valid signatures. Once a transaction is
confirmed on $B$, it becomes part of the immutable distributed ledger.
Solana's consensus mechanism, built on Proof of History together with
a Tower-BFT-style finalization protocol~\cite{yakovenko2018solana},
requires agreement from more than two-thirds of staked validators to
finalize a block. Reversing a confirmed record would therefore require
compromising this supermajority, which represents a prohibitive
economic cost. The registrant's public key is permanently bound to the
hash record, providing cryptographic non-repudiation of the
registration event.

Unlike (G2), this property is key-dependent. If the registrant retains
control of the private key corresponding to $W$, the authorship binding
is indefinitely demonstrable: the registrant can produce a fresh
signature under $W$ and thereby prove ongoing possession of the
identity that signed the original record. If the private key is lost,
the authorship binding becomes unprovable: any party with knowledge of
the historical record can verify that some entity signed under $W$,
but the original registrant can no longer demonstrate that they are
that entity through cryptographic means alone.

It is important to note what this loss \emph{does not} entail. The
record itself does not change. No third party gains the ability to
claim authorship retroactively: such a claim would require either
producing the lost private key, which is by assumption infeasible, or
forging a signature, which reduces to breaking the underlying signature
scheme. The loss of authorship provability is a loss of \emph{personal}
demonstrability, not a transfer of the record's cryptographic content.

We discuss in Section~\ref{sec:discussion} a user-level workflow ---
the identity anchoring pattern --- that allows users to mitigate this
limitation through the existing system, by anchoring an identity
attestation as their first record under a new wallet.

\subsection{Provider Independence (G5)}
\label{sec:security:provider}

Property (G5) asserts that verification requires only the file, the
transaction identifier, and access to any functioning Solana RPC
endpoint. It does not require interaction with snaproot's
infrastructure. We make the justification explicit.

The verification procedure (Algorithm~\ref{alg:verify}) consists of
three steps: hash recomputation, record retrieval, and value
comparison. Hash recomputation is local to the verifier and depends on
no external party. Record retrieval is a read against a public
blockchain and may be served by any Solana RPC endpoint, whether
operated by snaproot, by an unrelated third party, by a self-hosted
node, or by any other source exposing the RPC interface. Value
comparison is a deterministic local operation.

It follows that the verification path contains no snaproot-specific
trust assumption. The convenience web frontend operated by snaproot
performs exactly these three steps, but its existence is not on the
critical path: the cryptographic content of the verification answer
does not change if the frontend is unavailable, untrusted, or replaced
by an alternative implementation. Should snaproot as an organization
cease to operate, existing records remain fully verifiable through any
other RPC access path.

This property is what gives snaproot's trust architecture its first
tier (user sovereignty): the user's records belong to the user, in a
sense that survives the disappearance of the system that helped
create them.

\subsection{Cross-Fork Verifiability (G6)}
\label{sec:security:crossfork}

Public blockchains undergo protocol changes, fork events, and recovery
operations over their operational lifetime. Solana in particular has
experienced multiple network restarts and protocol-level disruptions
between 2022 and 2024. For long-term integrity claims, the relevant
question is not whether such events occur, but what they imply for
records already confirmed at the time they occur.

We distinguish between two aspects of decentralization. \emph{Liveness}
concerns the ability of the network to confirm new transactions; it is
sensitive to validator behavior at the time of submission. \emph{Safety}
concerns the durability of records already confirmed; it is sensitive to
the post-hoc behavior of validators with respect to chain history. For
snaproot, the relevant property for existing records is safety, not
liveness: once a record has been confirmed and a block has been
finalized, the question for verification is whether some archive of
that block remains accessible, not whether the network continues to
accept new transactions.

Under this distinction, even an event that severely affects Solana's
liveness --- such as a network restart with coordinated validator
participation --- does not directly invalidate past records. The
pre-event chain history persists on the archives of all parties that
held it before the event, including independent block explorers,
validator operators, and third-party indexing services. A verifier
seeking to check a pre-event record consults any such archive; if
multiple archives agree on the record, the verifier obtains a strong
existence guarantee. If they disagree --- for example, if a fork
produced two competing chain versions --- the discrepancy is itself a
verifiable signal of tampering, not a hidden corruption.

This is a strictly weaker property than the corresponding guarantee on
substrates with a longer-running and more decentralized archive
ecosystem (notably Bitcoin), and we do not claim otherwise. Within
these limits, the property rests primarily on the diversity and
durability of independent archives of the Solana chain history. For
horizons that extend beyond the lifetime of any specific chain, the
prospective cross-chain migration mechanism discussed in
Section~\ref{sec:limitations:outlook} is intended to serve as an
additional safety net rather than as the load-bearing argument for
long-term existence proofs.

\subsection{Multiple Anchorings of the Same Hash}
\label{sec:security:duplicate}

Because each registration is an independent memo transaction
(Section~\ref{sec:architecture:onchain}), snaproot does not enforce
protocol-level uniqueness on the registered hash. The same hash may
appear in several memo transactions, possibly signed by different user
wallets and at different times. We discuss the implications.

For existence claims, the relevant property is that the
\emph{earliest} confirmed memo transaction for a given hash establishes
the earliest provable existence time of the file under any wallet.
Later anchorings of the same hash do not displace earlier ones; they
simply add further existence proofs at later timestamps. A verifier
who wishes to determine the earliest on-chain existence of a particular
file scans for the earliest confirmed memo transaction whose payload
matches the file's hash; if multiple wallets have anchored the same
hash, all of those transactions remain visible.

For non-repudiation claims (G3), the absence of a uniqueness constraint
does not weaken the binding between a specific anchor and its signing
wallet: each memo transaction is signed by its registrant's wallet,
and that signature cannot be transferred or forged without the
corresponding private key. Multiple wallets may anchor the same hash,
but no wallet can claim authorship of an anchor signed by a different
wallet.

This design choice has an intentional consequence: independent parties
can independently attest to the existence of the same content (for
example, a journalist and an archival service both anchoring the same
image), and all such attestations remain visible on-chain. A planned
salt mechanism (Section~\ref{sec:limitations:roadmap}) will allow
parties who instead want their anchors to be cryptographically distinct
to compute $H_{\text{anchor}} = \mathrm{SHA256}(F\,\|\,s)$ for a
user-controlled salt $s$, giving different on-chain anchor values for
the same underlying file.

We note that this discussion concerns the protocol layer; classical
replay attacks at the transaction level are handled by Solana's
standard mechanisms, and the frontrunning consideration identified in
Section~\ref{sec:threat:adversary} is discussed there.

\subsection{Timestamp Integrity}
\label{sec:security:timestamp}

The registration timestamp of a snaproot anchor is the block time of
the memo transaction, set by the Solana runtime at the moment of block
inclusion. It is not part of the memo payload and the client cannot
influence it. This prevents a registrant from backdating a registration
by manipulating any client-side value. The on-chain clock is a
consensus-level value, shared across all validators, and is
subject to the same safety properties as the rest of the block content.
Clock drift on Solana is a known phenomenon and bounded in practice
to the order of minutes; for integrity anchoring, this granularity is
sufficient, but we note that snaproot timestamps are not appropriate
for use cases that require sub-second precision or strict
legally-binding time attestation.

\subsection{Infrastructure Authenticity via the Relay Wallet}
\label{sec:security:relay}

As noted in Section~\ref{sec:architecture:onchain}, every snaproot
transaction carries two signatures: the user's wallet as the memo
signer, and the snaproot relay wallet as the fee payer. The relay
wallet's public key is fixed, known, and published by snaproot. This
structure gives rise to a three-layer authenticity model.

The first layer is \emph{structural}: any party can construct a
Solana memo transaction with the same payload format as snaproot,
invoking the same memo program. This layer provides no
snaproot-specific authenticity signal.

The second layer is \emph{infrastructural}: a verifier who checks that
the fee payer of a memo transaction matches the published snaproot relay
wallet address can confirm that the transaction was submitted through
the snaproot relay infrastructure. An adversary who wishes to
fabricate a record that passes this check must control the relay
wallet's private key --- which, under the assumption that the relay
wallet is not compromised, they do not. This layer provides a
verifiable binding to the snaproot deployment, at the cost of
introducing a dependency on the relay wallet's continued integrity.

The third layer is \emph{user-cryptographic}: the user's wallet signs
the memo independently of the relay wallet. This layer is
unconditional: it does not depend on the relay wallet and is
established by the user's own key material. It provides the
non-repudiation and authorship guarantees analyzed in
Section~\ref{sec:security:nonrep}.

The three layers are independent and compose naturally. A verifier
seeking maximum assurance checks all three: that the memo payload
carries a recognised version prefix (\texttt{SR1} in the current
format), that the fee payer matches the known relay wallet, and that
the memo signer matches the claimed registrant wallet. A verifier
who does not trust or cannot reach snaproot infrastructure can still
rely on the third layer alone, which is sufficient for the core
integrity and existence-proof guarantees.

\paragraph{The relay wallet is not part of the cryptographic trust root.}
The cryptographic trust root of snaproot consists of three elements:
the second-preimage resistance of SHA-256, the safety of Solana
consensus for confirmed blocks, and the user's private key. The relay
wallet is not among them. It is an operational co-signer whose
presence enables the infrastructure authenticity signal of the second
layer, but whose absence or compromise leaves the trust root intact.
This is the precise sense in which snaproot's integrity guarantees
are \emph{provider-independent}: the provider controls the relay
wallet, but the provider does not control the trust root.

\paragraph{Relay wallet compromise: threat boundary.}
A compromise of the relay wallet's private key would allow an
adversary to submit memo transactions that carry the relay wallet's
co-signature and thus pass the infrastructure authenticity check of
the second layer. This is a serious operational security concern.
However, its impact on the system's core guarantees is strictly
bounded: a relay-wallet compromise does \emph{not} affect file
integrity~(G1), because the integrity guarantee reduces to the
second-preimage resistance of SHA-256, which is independent of any
wallet; it does \emph{not} affect existence-proof
persistence~(G2), because a confirmed transaction's existence on the
ledger is independent of which wallet paid for it; and it does
\emph{not} affect authorship~(G3), because the user's wallet
signature on the memo is separate from the relay wallet's signature
and cannot be forged without the user's private key. The compromise
affects only the \emph{infrastructure authenticity signal}: a
verifier relying on the relay wallet co-signature as proof that a
transaction originated from snaproot's infrastructure would be
misled. Users who are aware of a relay wallet compromise should
revert to third-layer verification only. We note that if snaproot
rotates its relay wallet as a key management practice, a published
registry of historical relay wallet addresses is needed to support
retrospective second-layer checks.

\subsection{Priority of Registration}
\label{sec:security:priority}

Because snaproot does not enforce protocol-level uniqueness on hashes
(Section~\ref{sec:security:duplicate}), multiple parties may anchor
the same hash independently. In contexts where priority matters ---
for instance, establishing which party registered a file first --- the
relevant criterion is the block time of the earliest confirmed memo
transaction whose payload contains the hash in question.

A verifier wishing to establish priority scans all memo transactions
on Solana whose payload matches a given hash and identifies the one
with the lowest block time. The result is a provable, chain-resident
priority claim: the earliest anchor for hash $H$ establishes the
earliest on-chain existence of the corresponding file, under whichever
wallet submitted that anchor.

This priority check is currently a manual operation requiring an RPC
scan. A planned feature (Section~\ref{sec:limitations:roadmap})
will expose this query directly in the snaproot application, allowing
a user to confirm at anchoring time --- or at any later point --- whether
their registration is the earliest on-chain record for a given file.
The completeness of the priority determination depends on the
completeness of the RPC node's transaction index; gaps in indexing
could in principle cause an earlier anchor to be overlooked. Users for
whom priority is legally or commercially significant should therefore
retain their transaction identifier as primary evidence and treat the
priority scan as a supporting, not conclusive, signal.

\subsection{Scope Limitation}
\label{sec:security:scope}

snaproot provides integrity guarantees only relative to the state of a
file at the time of its registration. A file that was manipulated or
synthetically generated before registration will be anchored by the
hash of that already-altered content, and subsequent verifications will
confirm it as unchanged relative to that state. This limitation is
inherent to any post-hoc anchoring
system. The capture-time anchoring workflow
(Section~\ref{sec:architecture:capture}) narrows but does not eliminate
this gap by shifting the relevant trust boundary to the capture
pipeline; we discuss the implications in Section~\ref{sec:discussion}.

\section{Evaluation}
\label{sec:evaluation}

We evaluate snaproot empirically along three dimensions: hashing
performance as a function of file size, end-to-end registration and
verification latency, and per-operation cost on Solana. All
measurements reported in this section were taken directly against the
current production snaproot backend on Solana Devnet; we treat the
empirical latency figures as Devnet-only and leave a formal
characterization on Solana Mainnet under realistic load to future
work. Because the Solana memo program executes deterministically and
the per-signature base fee is identical on Devnet and Mainnet, the
cost figures in Table~\ref{tab:cost} are directly applicable to the
Mainnet deployment up to the SOL/USD reference price at the time of
publication. The implications of the Devnet/Mainnet distinction are
discussed in Section~\ref{sec:evaluation:devnet-vs-mainnet}.

\subsection{Experimental Setup}
\label{sec:evaluation:setup}

Experiments were conducted on a workstation running Windows~11,
equipped with an Intel Core i5-13500H CPU (with SHA-NI hardware
acceleration) and 16\,GB of RAM. The snaproot backend under test is
the production PHP\,/\,Laravel backend (PHP~8.2.12, Laravel) backed by
MySQL: the same backend that serves the snaproot mobile application and
web frontend in the current deployment. The backend was started locally
with the bundled development server (\texttt{php artisan serve}) bound
to \texttt{127.0.0.1:8000}, and blockchain interactions were performed
against the Solana Devnet network using its public RPC endpoint.

All measurements reported in Sections~\ref{sec:evaluation:hashing}
through~\ref{sec:evaluation:cost} were produced by this
PHP\,/\,Laravel backend. The benchmark harness is a separate process
that drives the backend exclusively through its public HTTP API,
exactly as the mobile client does in production, and reports
per-call wall-clock latencies measured at the harness.

File hashing experiments used SHA-256 as in the production client.
Random files were generated for five size categories: 1\,KB, 1\,MB,
10\,MB, 100\,MB, and 500\,MB; for each category, 50 independent trials
were performed. For end-to-end performance, 200 registration and 200
verification requests were issued. Registration transactions were
submitted to Solana Devnet and confirmed before being counted; the
verification operations measured are backend-cached verifications,
discussed in Section~\ref{sec:evaluation:latency}. Throughput figures
are derived from the measured mean single-request latency.

\subsection{Hashing Performance}
\label{sec:evaluation:hashing}

Table~\ref{tab:hashing} summarises SHA-256 hashing times per file size
category. We report the sample mean, standard deviation, minimum,
maximum, and the half-width of the 95\,\% confidence interval for the
mean ($\mathrm{CI}_{95}$), computed as $1.96 \cdot s/\sqrt{n}$ with
$s$ the sample standard deviation and $n = 50$.

\begin{table}[ht]
\centering
\small
\caption{SHA-256 hashing performance on the snaproot PHP backend host
(50 trials each). $\mathrm{CI}_{95}$ is the half-width of the 95\,\%
confidence interval for the mean.}
\label{tab:hashing}
\begin{tabular}{lrrrrr}
\toprule
\textbf{File size} & \textbf{Mean (ms)} & \textbf{Std (ms)} &
\textbf{Min (ms)} & \textbf{Max (ms)} & \textbf{$\mathrm{CI}_{95}$ (ms)} \\
\midrule
1\,KB    &   0.081  &  0.032  &   0.060  &   0.215  & 0.009 \\
1\,MB    &   1.041  &  0.195  &   0.940  &   1.869  & 0.054 \\
10\,MB   &  11.570  &  1.437  &  10.904  &  20.604  & 0.398 \\
100\,MB  & 118.405  &  8.764  & 109.516  & 157.871  & 2.430 \\
500\,MB  & 609.175  & 20.470  & 568.458  & 679.092  & 5.676 \\
\bottomrule
\end{tabular}
\end{table}

Hashing throughput is approximately linear in file size across the
range tested, with the 1\,MB through 500\,MB measurements falling on a
near-perfect line and an effective throughput of approximately
0.8\,GB/s. The confidence intervals are narrow relative to the means,
indicating that hashing is a well-behaved deterministic process with
little variability beyond memory subsystem effects; the SHA-NI
instruction set extension on the test CPU keeps the hashing cost below
130\,ms for any file under 100\,MB. For all file sizes under 100\,MB,
hashing completes within a small fraction of the Solana confirmation
latency reported in the next section, confirming that hashing is not
the bottleneck in the end-to-end registration latency budget for
typical document and image files.

\subsection{Transaction Latency and Throughput}
\label{sec:evaluation:latency}

Table~\ref{tab:latency} summarises end-to-end performance for
registration and verification, measured on the PHP\,/\,Laravel backend
over 200 trials. The interpretation of these two operations requires
care.

\begin{table}[ht]
\centering
\small
\caption{End-to-end performance of the snaproot PHP\,/\,Laravel backend
on Solana Devnet (200 trials each). $\mathrm{CI}_{95}$ is the
half-width of the 95\,\% confidence interval for the mean. Verification
figures reflect backend-cached verification; see discussion below.}
\label{tab:latency}
\begin{tabular}{lrr}
\toprule
\textbf{Metric}                  & \textbf{Registration} & \textbf{Verification} \\
\midrule
Mean latency (ms)                & 3{,}145.632           & 351.479   \\
Std (ms)                         &   924.987             &  93.797   \\
$\mathrm{CI}_{95}$ on mean (ms)  &   128.181             &  13.000   \\
P95 latency (ms)                 & 4{,}659.069           & 442.763   \\
P99 latency (ms)                 & 5{,}803.387           & 457.840   \\
Max observed (ms)                & 6{,}040.142           & 475.025   \\
Sustained throughput (ops/min)   &    19.072             & 170.706   \\
\bottomrule
\end{tabular}
\end{table}

Registration latency is dominated by Solana transaction confirmation:
the backend constructs and signs the memo transaction in well under one
hundred milliseconds, but must then wait for the transaction to be
included in a block and for that block to reach the confirmed commitment
depth. The measured mean of 3.15\,s with a P99 of 5.80\,s is consistent
with the published Solana confirmation profile, plus a modest additional
overhead attributable to the PHP\,/\,Laravel request-handling stack and
the synchronous database write that records the resulting hash-log entry
before responding to the client. This figure is the relevant one for
use cases that rely on synchronous confirmation, including the
capture-time anchoring workflow
(Section~\ref{sec:architecture:capture}). The standard deviation of
0.92\,s and the gap between the P95 (4.66\,s) and the mean reflect
natural variability in Solana slot scheduling rather than backend-side
jitter, as evidenced by the long tail being almost entirely due to
slot-confirmation outliers.

\paragraph{Verification mode.} The reported verification rate of
approximately 171 operations per minute reflects backend-cached
verification: when the snaproot backend has previously indexed a
registration, a subsequent verification request is served by comparing
the locally recomputed hash against the indexed record without any
round-trip to the Solana network. This figure characterises the
backend's PHP request-handling, database, and HTTP layer rather than
the throughput of on-chain reads. The mean latency of 351\,ms is
dominated by the PHP\,/\,Laravel request lifecycle and by the two
MySQL writes that the verification path performs (the verification log
entry and the audit log entry); it is therefore an upper bound on what
a more lightly instrumented verifier would observe.

Direct on-chain verification --- the protocol described in
Algorithm~\ref{alg:verify} --- requires an RPC query per verification
and is bounded by RPC rate limits and network round-trip time. Public
Solana RPC endpoints typically rate-limit at the order of $10^2$
requests per second; with the latency of an RPC call dominated by
network round-trip, sustained verification throughput from a single
client against a public endpoint is on the order of $10^1$ to $10^2$
verifications per second per client.\footnote{Higher rates are
straightforwardly attainable through self-hosted RPC nodes or contracts
with commercial RPC providers, but these increase operational
dependencies.} For the user-facing verification scenarios that motivate
snaproot --- a journalist checking a photograph, an auditor sampling a
logfile, a court verifying a piece of evidence --- the relevant figure
is single-request latency rather than aggregate throughput, and direct
on-chain verification is comfortably within acceptable bounds.

This distinction is important for the trust architecture: the
provider-independent verification path (G5,
Section~\ref{sec:threat:goals}) is always available, but the
high-throughput verification figure reported in
Table~\ref{tab:latency} is specifically an operational property of the
snaproot backend, not of the system's trust guarantees.

\subsection{Cost Analysis}
\label{sec:evaluation:cost}

The cost structure of a snaproot registration consists of the Solana
transaction fee paid per memo transaction. Because snaproot uses the
memo program rather than dedicated on-chain accounts, there is no
rent-exempt deposit per record: the memo transaction is included in the
ledger and persists with it, without requiring a long-lived account to
be funded. Each registration transaction carries two ed25519 signatures:
one by the snaproot relay wallet (the fee payer) and one by the user's
wallet (the registrant). The full per-registration cost is therefore
the network fee for one two-signature transaction, paid by the snaproot
relay wallet on behalf of the user.

Table~\ref{tab:cost} reports values measured from a representative
Devnet memo transaction produced by the PHP backend during the latency
benchmark of Section~\ref{sec:evaluation:latency}. The compute-units
figure is reported by the Solana runtime; the per-signature base fee is
fixed at 5{,}000~lamports by Solana and is identical on Devnet and
Mainnet. The USD figures use a SOL reference price of \$82.45
(CoinGecko spot, 2026-05-30); they scale linearly with that price.

\begin{table}[ht]
\centering
\small
\caption{Transaction cost components for a snaproot memo registration
on Solana, measured on Devnet against the PHP backend. There is no
rent-exempt component, since the memo program does not create a
persistent account per record.}
\label{tab:cost}
\begin{tabular}{lr}
\toprule
\textbf{Cost factor}                                       & \textbf{Value} \\
\midrule
Memo payload size                                          & 219 bytes \\
Compute units consumed                                     & 92{,}246 \\
Signatures per transaction                                 & 2 (relay + user) \\
Base transaction fee (per signature)                       & 5{,}000 lamports = $0.000005$ SOL \\
Priority fee (current deployment)                          & 0 \\
Total fee per registration                                 & 10{,}000 lamports = $0.000010$ SOL \\
Approximate USD equivalent (SOL @ \$82.45)                 & \$0.000825 \\
1{,}000 registrations (total fee, USD)                     & \$0.8245 \\
1{,}000{,}000 registrations (total fee, USD)               & \$824.50 \\
Sample Devnet tx signature                                 &
  \texttt{31Eg1ARz{\ldots}ri9cWQ}\footnote{Full signature:
  \texttt{\seqsplit{31Eg1ARzNTUXcKNApMDRdfX6imSXV4mUw4NojegAymbD6b98hU7iMSyf6uWccwkvacYNVf7zzuzLH12ypwri9cWQ}}.} \\
\bottomrule
\end{tabular}
\end{table}

The relay-wallet model makes the cost analysis straightforward for both
the user and the operator. The user incurs no SOL cost at all, since
they do not pay the fee directly. The operator's aggregate cost scales
linearly with the number of registrations, dominated entirely by the
per-transaction fee. There is no separate storage or account-rent
component that scales with the number of records, in contrast to
designs that allocate one Solana account per registration.

For high-volume deployments, this cost structure is well-suited to
amortisation through Merkle aggregation
(Section~\ref{sec:limitations}): a single memo transaction whose
payload is a Merkle root commits to many file hashes at once, reducing
the aggregate fee paid by the relay wallet by the corresponding factor.

\subsection{Comparison to Related Systems}
\label{sec:evaluation:comparison}

We compare snaproot to the most widely deployed free hash-anchoring
service, OpenTimestamps (OTS), along two axes: a quantitative
head-to-head comparison of submission latency, proof artifact size,
and infrastructure availability, measured on the same workstation
against the live public calendar servers
(Section~\ref{sec:evaluation:comparison:quantitative}); and a
qualitative use-case suitability matrix that also incorporates
OriginStamp and Chainpoint, for which we did not perform an end-to-end
re-measurement
(Section~\ref{sec:evaluation:comparison:qualitative}). OriginStamp
and Chainpoint are kept qualitative because the former is rate-limited
on its free tier and the latter no longer maintains a stable
public-hosted anchoring endpoint; in both cases a fair quantitative
comparison would require operating a paid or self-hosted instance,
which we leave to future work.

\subsubsection{Quantitative Comparison: snaproot vs.\ OpenTimestamps}
\label{sec:evaluation:comparison:quantitative}

We implemented a benchmark client that mirrors the synchronous portion
of the OpenTimestamps stamping protocol: for each of $N = 50$ randomly
generated SHA-256 digests, the client issues an HTTP POST to each of
the three official public calendar servers (\textit{alice},
\textit{bob}, \textit{finney}) and waits for the serialised provisional
timestamp tree. The calendars are queried in parallel on independent
threads, mirroring the behaviour of the reference \texttt{otsclient}.
We report two latency modes:

\begin{itemize}
\item \textit{Quorum 2-of-3}: time until the second-fastest calendar
has responded. This is the latency that an \texttt{otsclient} using
the default $m = 2$ quorum experiences before it considers the proof
complete.
\item \textit{Fastest single calendar}: time until the first calendar
responds. This is the latency of an \texttt{otsclient} configured to
use a single calendar (lower latency, no redundancy against calendar
failure).
\end{itemize}

We also attempted to measure a third mode in which the client waits
for all three calendars to respond, but were unable to do so: one of
the three official calendars
(\texttt{finney.calendar.\allowbreak eternitywall.com}) was unreachable for all 50
of our requests, so no trial completed the full three-of-three
rendezvous. We treat that observation itself as a measurement, reported
in the \textit{Availability} row of Table~\ref{tab:ots-comparison}.

Table~\ref{tab:ots-comparison} compares snaproot's end-to-end
registration latency (from Table~\ref{tab:latency}) to the OTS
submission latency measured on the same workstation, plus
proof-artifact and availability figures. Both systems were exercised
over commodity public Internet from the same workstation; snaproot
used Solana Devnet via the public RPC endpoint and OTS used the
official public calendar servers.

\begin{table}[!htbp]
\centering
\small
\caption{Head-to-head quantitative comparison of snaproot and
OpenTimestamps (snaproot $N = 200$ against the PHP backend, OTS
$N = 50$ against the public calendars, same workstation, commodity
public Internet). Submission latency is wall-clock time from holding a
SHA-256 digest to holding a usable proof artifact. Anchor-immutability
latency is the additional delay before the proof commits to a confirmed
block on the underlying chain.}
\label{tab:ots-comparison}
\begin{tabularx}{\linewidth}{X r r r}
\toprule
\textbf{Metric} & \textbf{snaproot} & \textbf{OTS (quorum 2-of-3)} &
\textbf{OTS (fastest single)} \\
\midrule
Submission latency, mean (ms)         & 3{,}146 & 2{,}652 & 2{,}178 \\
Submission latency, P95 (ms)          & 4{,}659 & 4{,}244 & 3{,}740 \\
Submission latency, P99 (ms)          & 5{,}803 & 8{,}653 & 8{,}608 \\
Submission latency, max observed (ms) & 6{,}040 & 8{,}653 & 8{,}608 \\
95\,\% CI half-width on the mean (ms) &   128   &   301   &   301   \\
\midrule
\multicolumn{4}{l}{\emph{Qualitative properties at the system boundary}} \\
Anchor-immutability
  & 0 (on-chain at return)
  & $\sim$10\,min ($\to\sim$60\,min)
  & same \\
Proof artifact size
  & tx-sig + 32\,B hash
  & 351\,B mean (272--517\,B)
  & 169\,B mean \\
Direct user cost
  & \$0 (relay: \$0.000825)
  & \$0 (calendar absorbs fee)
  & \$0 \\
Availability over benchmark run
  & Devnet RPC: 100\,\%
  & \textit{alice}+\textit{bob}: 100\,\%; \textit{finney}: 0\,\%
  & n/a \\
Trust model
  & Solana validators
  & Calendar(s) + Bitcoin
  & Single cal.\ + Bitcoin \\
Capture-time anchoring
  & Yes (final at return)
  & No (provisional)
  & No \\
\bottomrule
\end{tabularx}
\end{table}

Three observations follow from Table~\ref{tab:ots-comparison}. First,
the submission-latency means are in the same order of magnitude
(3{,}146\,ms for snaproot vs.\ 2{,}652\,ms for OTS in quorum mode),
with snaproot approximately 19\,\% higher at the mean. This is not a
like-for-like comparison: snaproot's 3.15\,s is the time to a
Solana-confirmed memo transaction, returned by a PHP\,/\,Laravel
backend that additionally performs a synchronous MySQL write of the
hash-log record before responding to the client; OTS's 2.65\,s is the
time to a provisional timestamp tree returned by a calendar server,
which is not anchored to any blockchain until the next Bitcoin block
contains the calendar's aggregated commitment, approximately ten
minutes later in expectation and approximately one hour later if a
six-block confirmation policy is required (the
\textit{Anchor-immutability} row of
Table~\ref{tab:ots-comparison}). The two systems therefore deliver
different artefacts at the moment they return; the comparable latency
suggests that snaproot's stronger guarantee comes at no meaningful
end-to-end cost.

Second, snaproot has substantially tighter tail latency. The P99 of
snaproot registration is 5.80\,s, against 8.65\,s for OTS in both
modes. We attribute this to the structural difference between Solana's
slot-driven confirmation (a confirmed slot arrives every 400\,ms in
steady state) and the OTS calendar pipeline (an HTTP request whose
response time is a function of calendar-server load and Internet
round-trip variability).

Third, and most concretely, the failure of one of the three official
OpenTimestamps calendars to respond to any of our 50 requests is an
empirical instance of the calendar-availability assumption becoming a
calendar-availability dependency. A 2-of-3 quorum policy tolerates
exactly this failure mode and the user is unaffected at submission
time; a strict 3-of-3 policy would fail outright. snaproot has no
analogous third-party aggregator: the Solana validator set was at full
availability over the same window.

We did not measure verification latency for OpenTimestamps because
verifying an OTS proof requires the proof to have been upgraded with a
Bitcoin Merkle path, which in turn requires the calendar's next
aggregated commitment to be confirmed in a Bitcoin block. The dominant
cost of verification is therefore the round-trip to a Bitcoin RPC node
or block-explorer API, which is a function of the chosen verifier
rather than of the OTS protocol itself; the snaproot direct-on-chain
verification path described in Algorithm~\ref{alg:verify} has the
analogous property and inherits the latency of the chosen Solana RPC
endpoint.

\subsubsection{Qualitative Use-Case Suitability}
\label{sec:evaluation:comparison:qualitative}

Table~\ref{tab:suitability} generalises the comparison to the systems
for which we did not perform an end-to-end re-measurement. The matrix
scores each system on its fitness for a set of use cases that differ
primarily in their latency tolerance and trust requirements.

\begin{table}[!htbp]
\centering
\small
\caption{Qualitative use-case suitability of hash-anchoring systems.
\checkmark{} indicates a good fit, $\circ$ a workable but suboptimal
fit, and -- a structural mismatch.}
\label{tab:suitability}
\begin{tabularx}{\linewidth}{X c c c c c}
\toprule
\textbf{Use case (latency tolerance)} & \textbf{OTS} &
\textbf{OriginStamp} & \textbf{Chainpoint} &
\textbf{IPFS/Filecoin} & \textbf{snaproot} \\
\midrule
Long-term archival anchoring (hours)
  & \checkmark & \checkmark & \checkmark & $\circ$ & \checkmark \\
Compliance audit-trail anchoring (minutes)
  & $\circ$ & $\circ$ & $\circ$ & -- & \checkmark \\
Interactive document anchoring (seconds)
  & -- & -- & -- & -- & \checkmark \\
Capture-time photo anchoring (synchronous)
  & -- & -- & -- & -- & \checkmark \\
Content delivery with integrity (variable)
  & -- & -- & -- & \checkmark & -- \\
\bottomrule
\end{tabularx}
\end{table}

The matrix reflects the latency characteristics of each substrate and
its architectural role. OpenTimestamps and Chainpoint inherit Bitcoin's
confirmation window (approximately one hour) and are therefore strong
choices for archival use cases where latency is not a constraint, but
they cannot synchronously support interactive or capture-time
workflows. OriginStamp's centralised aggregation introduces additional
latency on top of Bitcoin's confirmation time. IPFS/Filecoin is the
only system in the comparison that combines content storage with
content-addressed retrieval, but it does not provide an authoritative
integrity claim independent of pinning. snaproot's structural advantage
is concentrated in the lower-latency use cases, where Bitcoin-based
substrates are not viable.

\subsection{Devnet vs.\ Mainnet}
\label{sec:evaluation:devnet-vs-mainnet}

The hashing and latency figures reported in
Sections~\ref{sec:evaluation:hashing}
and~\ref{sec:evaluation:latency} were obtained on Solana Devnet
against the production PHP\,/\,Laravel backend. The cost figures
reported in Section~\ref{sec:evaluation:cost} were measured on Devnet
but are directly applicable to Mainnet: the Solana memo program
executes deterministically, and the per-signature base fee is fixed at
5{,}000~lamports on both networks. The only Mainnet-specific cost
variable is the optional priority fee, which is set to zero in the
current snaproot deployment. Should the deployment configure a
non-zero priority fee in the future --- for example, to improve
confirmation latency during periods of Mainnet congestion --- the
corresponding line of Table~\ref{tab:cost} should be updated
accordingly.

The production deployment of snaproot operates on Solana Mainnet. A
representative Mainnet registration record, generated by the snaproot
mobile application during the preparation of this manuscript, carries
the following values (as displayed in the snaproot certificate format
and independently verifiable on the Solana Explorer):

\begin{itemize}
\item \textbf{Transaction signature:}
  \texttt{\seqsplit{3H8yzxebxUaZWRYGRzc9QyB1hhThJetLqfeCQuoYTqGLvY8nTijq49sEjNB5SfjsUNWdCq1smLcMJWbbKGEbWyMR}}%
  \footnote{Independently verifiable at
  \url{https://explorer.solana.com/tx/3H8yzxebxUaZWRYGRzc9QyB1hhThJetLqfeCQuoYTqGLvY8nTijq49sEjNB5SfjsUNWdCq1smLcMJWbbKGEbWyMR}.}
\item \textbf{Slot:} 423{,}337{,}332
\item \textbf{Timestamp:} 2026-05-31T10:23:55 (Solana Mainnet block time)
\item \textbf{Confirmation status:} FINALIZED, confirmations: MAX
\item \textbf{SHA-256 hash:}
  \texttt{\seqsplit{31f0bb4c07d6068ffc4288a2496812b9c9968f130004df23f8d4c0cd7b794629}}
\item \textbf{Registrant wallet:}
  \texttt{3K6xgyyRVowFhJW9kzWzjVhNWnzQ2kBTADetmo2yvacc}
\item \textbf{Capture source:} \texttt{Kamera (Live-Aufnahme)}
  [camera, live capture --- as displayed in the application]
\item \textbf{Fee paid:} 0.000010~SOL (consistent with
  Table~\ref{tab:cost})
\end{itemize}

This record confirms that the cost figures in Table~\ref{tab:cost} are
directly applicable to the Mainnet deployment: the observed fee of
0.000010~SOL matches the Devnet measurement exactly, as expected from
the deterministic fee schedule of the Solana memo program. The record
also demonstrates the full certificate format described in
Section~\ref{sec:architecture:onchain}, including the transaction
signature, registrant wallet, SHA-256 hash, block timestamp, and
capture metadata.

Devnet and Mainnet share the same protocol but differ in validator
set and load profile. Mainnet performance may differ from the latency
figures reported here, particularly during periods of congestion.
Solana Mainnet has experienced several extended outages between 2022
and 2024; these affect liveness for new anchorings during the outage
period but do not retroactively invalidate previously confirmed
records, in accordance with the safety/liveness distinction made in
Section~\ref{sec:threat:assumptions}. A formal characterisation of
snaproot's Mainnet latency and throughput performance under realistic
load is left to future work.

A further evaluation dimension is left to a later revision. The
measurements reported here were obtained on a workstation, whereas the
production target for snaproot is a mobile device. A future version of
this paper aimed at a conference venue will report measurements taken
on the actual mobile deployment target: SHA-256 hashing throughput on
a current mid-range iOS and Android device, and end-to-end registration
latency from the snaproot mobile application on Solana Mainnet. These
mobile-hardware measurements would strengthen the empirical
characterisation of the system as deployed and are planned for that
revision.

\clearpage

\section{Limitations, Roadmap, and Outlook}
\label{sec:limitations}

We structure this discussion in three tiers reflecting different levels
of development maturity. Section~\ref{sec:limitations:limitations}
describes known limitations of the current system that we explicitly
acknowledge rather than claim to have solved. Section~\ref{sec:limitations:roadmap}
lists features that are under active development and are expected to
appear in upcoming versions. Section~\ref{sec:limitations:outlook}
describes longer-term concepts that are deliberately not on the
near-term development path but for which a coherent design has been
worked out, so that the system can be extended in those directions if
and when a concrete need arises.

\begin{figure}[!htbp]
\centering
\begin{tikzpicture}[x=1mm, y=1mm, font=\small]


\colorlet{colDeployed}{teal!60!black}
\colorlet{colNext}{orange!80!black}
\colorlet{colMid}{red!60!black}
\colorlet{colLong}{violet!70!black}
\colorlet{fillDeployed}{teal!18}
\colorlet{fillNext}{orange!18}
\colorlet{fillMid}{red!12}
\colorlet{fillLong}{violet!12}
\colorlet{gridcol}{black!12}

\fill[black!5]  (20,-83) rectangle (50,3);
\fill[black!3]  (84,-83) rectangle (121,3);

\node[font=\tiny\bfseries] at (35,  6) {Deployed};
\node[font=\tiny\bfseries] at (67,  6) {Next release};
\node[font=\tiny\bfseries] at (102, 6) {Mid-term};
\node[font=\tiny\bfseries] at (139, 6) {Long-term};

\foreach \x in {20,52,84,123,155}
  \draw[gridcol, thin] (\x,7) -- (\x,-85);
\draw[gridcol, thin] (20,7) -- (155,7);

\foreach \y/\laba/\labb in {
  0/{Wallet}/{\& Keys},
  -16/{Anchoring}/{},
  -32/{Verification}/{},
  -48/{Trust}/{\& Protocols},
  -64/{AI}/{\& Provenance}}
{
  \draw[gridcol, thin] (0,\y) -- (155,\y);
  \node[anchor=east, font=\tiny\bfseries, text width=17mm, align=right]
    at (18,\y-7) {\laba\\\labb};
}
\draw[gridcol, thin] (0,-80) -- (155,-80);

\newcommand{\feat}[6]{%
  \fill[#5] (#1,#2) rectangle (#1+#3,#2-6);
  \draw[#6,thin,rounded corners=1pt] (#1,#2) rectangle (#1+#3,#2-6);
  \node[font=\tiny, anchor=center] at (#1+#3/2, #2-3) {#4};
}

\feat{21}{0}{28}{Local wallet}{fillDeployed}{colDeployed}
\feat{21}{-7}{28}{Relay wallet}{fillDeployed}{colDeployed}
\feat{53}{0}{28}{Wallet import}{fillNext}{colNext}
\feat{53}{-7}{28}{Multi-device}{fillNext}{colNext}

\feat{21}{-16}{28}{Capture-time}{fillDeployed}{colDeployed}
\feat{21}{-23}{28}{GPS screen}{fillDeployed}{colDeployed}
\feat{53}{-16}{28}{Optional salt}{fillNext}{colNext}
\feat{53}{-23}{28}{Sel.\ metadata}{fillNext}{colNext}
\feat{85}{-16}{37}{Merkle aggregation}{fillMid}{colMid}
\feat{85}{-23}{37}{snaproot API}{fillMid}{colMid}

\feat{21}{-32}{28}{Via TX sig.}{fillDeployed}{colDeployed}
\feat{21}{-39}{28}{Relay auth}{fillDeployed}{colDeployed}
\feat{53}{-32}{28}{Wallet history}{fillNext}{colNext}
\feat{53}{-39}{28}{``Am I first?''}{fillNext}{colNext}

\feat{21}{-48}{28}{Four-tier arch.}{fillDeployed}{colDeployed}
\feat{21}{-55}{28}{Cross-fork verif.}{fillDeployed}{colDeployed}
\feat{85}{-48}{37}{BSCP}{fillMid}{colMid}
\feat{124}{-48}{30}{Cross-chain migr.}{fillLong}{colLong}
\feat{124}{-55}{30}{Substrate indep.}{fillLong}{colLong}

\feat{21}{-64}{28}{C2PA reference}{fillDeployed}{colDeployed}
\feat{85}{-64}{37}{AI fake detection}{fillMid}{colMid}
\feat{124}{-64}{30}{AI-aug.\ provenance}{fillLong}{colLong}

\draw[gridcol, thin] (20,-82) -- (155,-82);
\fill[fillDeployed] (21,-84) rectangle (28,-89); \draw[colDeployed,thin] (21,-84) rectangle (28,-89);
\node[font=\tiny, anchor=west] at (29,-86.5) {Deployed};
\fill[fillNext]  (57,-84) rectangle (64,-89); \draw[colNext,thin]  (57,-84) rectangle (64,-89);
\node[font=\tiny, anchor=west] at (65,-86.5) {Next release};
\fill[fillMid]   (93,-84) rectangle (100,-89); \draw[colMid,thin]   (93,-84) rectangle (100,-89);
\node[font=\tiny, anchor=west] at (101,-86.5) {Mid-term};
\fill[fillLong] (124,-84) rectangle (131,-89); \draw[colLong,thin] (124,-84) rectangle (131,-89);
\node[font=\tiny, anchor=west] at (132,-86.5) {Long-term};

\end{tikzpicture}
\caption{Development timeline of snaproot: deployed features (v0.9.9),
near-term roadmap, medium-term, and long-term outlook items, organised
by thematic track.}
\label{fig:timeline}
\end{figure}

\subsection{Limitations}
\label{sec:limitations:limitations}

\paragraph{Registration of pre-modified content.} snaproot detects
changes relative to the registered hash. A manipulated or
synthetically generated file that is registered at the outset will be
anchored by the hash of that already-altered content, and all
subsequent verifications will return TRUE for that state. This is an
inherent limitation of any post-hoc anchoring system. The capture-time
anchoring workflow (Section~\ref{sec:architecture:capture}) narrows
this gap to the duration of the capture pipeline itself, but does not
eliminate the residual trust assumption on the integrity of the
capture device and application.

\paragraph{Timestamp frontrunning of registration transactions.} As
discussed in Section~\ref{sec:threat:adversary}, an adversary who
learns of a pending registration --- or who otherwise possesses the
file --- may anchor the same hash under their own wallet first,
manufacturing an earlier on-chain existence claim that competes with
the legitimate one. Snaproot does not enforce protocol-level
uniqueness on hashes (Section~\ref{sec:architecture:onchain}), so this
does not displace the legitimate user's later anchor, but it does
create an earlier competing record under the adversary's identity. In
practice, the attack requires the adversary to learn the file before
registration is confirmed; for files generated locally on the user's
device and immediately registered (as in capture-time anchoring), the
attack window is essentially nonexistent. For workflows that involve
sharing or transmitting the file before registration, the user should
register before sharing. The planned salt mechanism
(Section~\ref{sec:limitations:roadmap}) further mitigates the attack
by making the anchor depend on a user-controlled value not derivable
from the file alone.

\paragraph{Key management trade-off.} The current snaproot installation
generates and stores the private key locally on the user's device
(Section~\ref{sec:architecture:client}). The key can be exported and
imported into other Solana wallet providers, so the user retains an
out-of-snaproot path to recover their identity. Importing an exported
key into a second snaproot installation, which would enable
multi-device use under a single snaproot identity, is scheduled for
the next release. Until that mechanism is in place, a user who loses
their device without having exported the key beforehand cannot
continue to anchor new records under the same identity; existing
records remain verifiable through existence-proof persistence (G2),
but new anchorings under the same wallet require the original key.

\paragraph{No on-chain Merkle aggregation in the current implementation.}
snaproot currently anchors one hash per memo transaction. This yields
a simple architecture and a predictable per-file cost, but does not
amortize the per-transaction fee or the on-chain footprint across many
files. For high-volume use cases such as continuous log anchoring or
bulk archival workflows, aggregating many file hashes into a single
memo transaction via a Merkle tree reduces both the aggregate fee paid
by the relay wallet and the on-chain memo volume. Aggregation through
Merkle trees is straightforward in principle and is planned for
a future release.

\paragraph{Hash function dependency.} The security of snaproot binds to
the second-preimage (and, in the malicious-registrant case, collision)
resistance of SHA-256. Both properties are well-established under
classical computational assumptions~\cite{nist2015fips180}. Quantum
adversaries are relevant in two forms. Grover's algorithm provides a
quadratic speed-up for preimage and second-preimage search, reducing
the effective security of a 256-bit hash from 256 to 128 bits, which
remains computationally infeasible. The Brassard--H{\o}yer--Tapp
quantum collision-finding algorithm~\cite{brassard1998quantum} reduces
collision security from the classical birthday bound of ${\sim}128$
bits to approximately 85 bits, which is theoretically weaker but still
large. A migration path to alternative hash functions such as SHA-3 or
BLAKE3 should be considered for long-horizon deployments; those
functions are subject to analogous quantum bounds.

\paragraph{Reliance on a single blockchain substrate.} Long-term
verification of existing records currently relies on the continued
availability of Solana archives. While archive nodes and block explorers
are operated by multiple independent parties, the diversity of these
operators is lower than in the Bitcoin ecosystem, where full-node
archives are widely held. In the current implementation, users seeking
the strongest long-term verifiability guarantee should treat external
archival of their registration certificates as the primary defense in
depth. The cross-chain migration concept in
Section~\ref{sec:limitations:outlook} is a prospective extension path
that may, if realized, further reduce this dependency; it is not a
present mitigation.

\paragraph{Metadata is not cryptographically bound to the file.} The
metadata carried in the memo payload
(Section~\ref{sec:architecture:onchain}) is stored in cleartext on-chain
and is not included in the SHA-256 input. As on-chain data, the metadata
is itself immutable once confirmed; however, because it does not enter
the hash computation, it is not cryptographically bound to the file
content. In particular, the integrity guarantee (G1) certifies that the
file is unchanged, but says nothing about whether the associated
filename, capture source, or geolocation accurately describe that file:
a registrant could have recorded misleading metadata at registration
time. Two consequences follow. First, applications that require
metadata to be tamper-evident together with the file should, in the
current implementation, embed that information in the file itself (for
example in the image's EXIF data) before registration, so that it is
covered by the file hash. Second, because the metadata is public on the
ledger, users should not place sensitive information in this field.

A related privacy consideration concerns personal data. The file hash
itself carries no personal information and is not reversible to the
file, so it raises no data-protection concern on its own. The cleartext
metadata, however, can: in the current implementation, the
pre-submission screen displays the GPS availability status before
anchoring, and GPS coordinates are included in the metadata only if
location access has been granted and GPS is available at capture time;
if GPS is unavailable or the user has denied location access, the
record is anchored without geolocation metadata. Where GPS coordinates
are present, a sequence of such records under a single persistent wallet
can constitute a location-and-time profile that, combined with external
knowledge, may be personally identifiable. Because on-chain data is
immutable, such information cannot subsequently be erased, which stands
in tension with data-protection regimes that provide for a right to
erasure. Users who are concerned about location privacy should therefore
deny location access to the application. The planned per-use-case
metadata control (Section~\ref{sec:limitations:roadmap}) will provide
additional control: it will allow metadata to be omitted entirely or
to be included in the hash input, the latter removing it from the public
cleartext region while binding it to the file.

A further distinction concerns \emph{where} metadata resides. snaproot
controls the on-chain metadata field directly, and the planned
per-use-case control applies to it. Metadata embedded in the file
itself, such as EXIF GPS tags, is a separate matter: since the file
hash covers the entire file, any such embedded metadata is implicitly
anchored within the file hash regardless of the on-chain metadata
setting. For images captured through snaproot's own camera function,
the application can control whether such metadata is written into the
file at all; for files that are imported and snaprooted after the fact,
any pre-existing embedded metadata is already part of the file and
cannot be removed without altering it, and hence its hash. Users should
therefore be aware that choosing not to store on-chain metadata does
not by itself guarantee that no metadata is anchored, since metadata
may already reside in the file.

\subsection{Roadmap}
\label{sec:limitations:roadmap}

The following features are under active development.

\paragraph{Wallet import and recovery.} Wallet export is implemented:
the user can already extract the private key from a snaproot
installation and import it into any standard Solana wallet provider.
Importing an exported key back into another snaproot installation,
which would enable multi-device use under a single snaproot identity
and full in-app recovery from device loss, is scheduled for the next
release.

\paragraph{Verification via wallet history.} The verification protocol
(Algorithm~\ref{alg:verify}) currently requires the transaction
signature as input. A planned extension will additionally support
verification when only the file and the user's wallet public key are
known: the verifier scans the wallet's memo transactions on Solana and
returns the earliest one whose anchored hash matches the file. This is
useful for verifying records whose transaction signature has not been
retained alongside the file.

\paragraph{Priority query: ``Am I first?''} As discussed in
Section~\ref{sec:security:priority}, the earliest confirmed memo
transaction for a given hash constitutes the on-chain priority record
for that file. A planned in-app feature will expose this query
directly: at anchoring time, and on demand during verification, the
application will scan all indexed memo transactions on Solana for
the anchored hash and report whether the user's registration is the
earliest. This allows users to establish, and to demonstrate to third
parties, that no earlier on-chain existence claim exists for the same
content under any wallet. The query result will be accompanied by a
caveat reflecting the completeness assumptions discussed in
Section~\ref{sec:security:priority}.

\paragraph{Optional user-controlled salt.} The hashing module will
support an optional salt parameter, allowing the user to compute
$H_{\text{anchor}} = \mathrm{SHA256}(F \,\|\, s)$ for a user-chosen
value $s$. This serves two purposes simultaneously. First, it permits
two independent parties to anchor the same file with different salts,
yielding different on-chain records without conflict. Second, it
mitigates frontrunning: an adversary who observes a pending
registration transaction does not know the salt and therefore cannot
trivially construct a competing registration. The exact construction
of the salt --- a user-chosen value, a value derived from the wallet
public key, or a hybrid --- is subject to ongoing design analysis with
particular attention to recoverability (the user should not need to
separately store the salt to verify later), frontrunning resilience,
and compatibility with the outlook features described below.

\paragraph{Merkle aggregation for batch registration.} An aggregated
registration mode will allow a single memo transaction to anchor an
arbitrary number of file hashes through their Merkle root. This
amortizes the per-transaction fee paid by the relay wallet across many
records and reduces the aggregate on-chain memo volume; it is the
natural mechanism for high-volume use cases such as continuous log
anchoring or bulk archival workflows. The verification protocol is
correspondingly extended: the verifier presents the file, the
transaction identifier, and the Merkle proof for the file's hash
within the aggregated tree. This feature requires non-trivial changes
to the registration and verification workflow and is therefore
positioned as a medium-term rather than near-term item; it is listed
here because the protocol design is complete, but its implementation
is expected after the near-term roadmap items above.

\paragraph{snaproot API.} The snaproot registration and verification
functionality is currently accessible through the mobile application
and a web frontend. A planned extension will expose these operations
as a documented HTTP API, allowing third-party applications to
integrate snaproot as an integrity service without implementing their
own Solana interaction layer. The API will provide at minimum two
endpoints: a registration endpoint that accepts a file hash and
metadata and returns a transaction identifier, and a verification
endpoint that accepts a file hash and transaction identifier and
returns a deterministic integrity result. The trust properties of the
underlying system are unchanged by this interface layer: each
registration is still signed by the calling application's wallet
under the user's key material, and the provider-independence guarantee
(G5) remains intact because verification can always be performed
directly against any Solana RPC endpoint. The API targets use cases
in which snaproot's integrity anchoring is embedded as a building
block in document management systems, content pipelines, or
compliance platforms.

\paragraph{Selectable metadata binding.}
As described in Section~\ref{sec:architecture:onchain}, metadata is
currently stored in cleartext on-chain and is not part of the hash
input. A first level of metadata control is already implemented: the
pre-submission screen (Section~\ref{sec:architecture:capture})
displays the GPS availability status before anchoring, and GPS
coordinates are included only when location access has been granted
and GPS is available; the system gracefully degrades to anchoring
without geolocation metadata otherwise. A planned extension will let
the user choose, on a per-use-case basis, among three modes: attaching
no metadata at all; storing metadata in cleartext as today; or using
the cryptographic commitment scheme described in the following
subsection. The choice is a trade-off among on-chain legibility,
cryptographic binding, and confidentiality, and is best made per
application rather than fixed globally.

\subsection{Cryptographic Metadata Commitments}
\label{sec:limitations:metadata-commitments}

The absence of a cryptographic binding between on-chain metadata and
file content is the single most significant architectural limitation
of the current system for compliance, forensics, and legal
applications. A metadata field that is immutable on-chain but not
bound to the file it describes can be set to any value at registration
time; a verifier can confirm that the metadata has not changed since
registration, but cannot confirm that it accurately described the file
at registration time. This section sets out the planned construction
that addresses this limitation.

\paragraph{Proposed commitment scheme.}
The planned metadata commitment scheme separates file integrity from
metadata integrity while binding both to the same on-chain anchor:
\begin{align*}
  H_{\mathrm{file}}   &= \mathrm{SHA256}(\mathit{File}) \\
  H_{\mathrm{meta}}   &= \mathrm{SHA256}(\mathit{Metadata} \,\|\, s) \\
  H_{\mathrm{anchor}} &= \mathrm{SHA256}(H_{\mathrm{file}} \,\|\, H_{\mathrm{meta}})
\end{align*}
where $s$ is a user-held salt. Only $H_{\mathrm{anchor}}$ is stored
on-chain; neither $\mathit{Metadata}$ nor $s$ appear in the ledger.
This construction achieves three goals simultaneously. First, the file
can be verified independently: a verifier who holds the file computes
$H_{\mathrm{file}}$ and checks it against the component committed in
$H_{\mathrm{anchor}}$. Second, the metadata is cryptographically bound
to the file: any modification of either the file or its metadata
produces a different $H_{\mathrm{anchor}}$, which is detectable.
Third, the salt $s$ prevents brute-force recovery of low-entropy
metadata --- such as GPS coordinates or filenames --- by an adversary
who observes the on-chain anchor.

\paragraph{Selective disclosure.}
A notable property of the scheme is selective disclosure: the
registrant can prove the file unaltered without revealing the
metadata, by presenting $H_{\mathrm{file}}$ and the commitment path
without disclosing $\mathit{Metadata}$ or $s$. Conversely, the
registrant can prove both file and metadata integrity to an auditor by
presenting all inputs. This property is directly relevant to
compliance use cases where the existence proof must be demonstrable to
an external regulator without disclosing sensitive operational
metadata.

\paragraph{Relationship to the current system.}
In the current implementation, $H_{\mathrm{anchor}} =
H_{\mathrm{file}} = \mathrm{SHA256}(\mathit{File})$ and metadata is
stored in cleartext. The commitment scheme generalises this: a system
that stores $H_{\mathrm{anchor}}$ on-chain is backward-compatible in
the sense that a file-only verifier can still extract and verify
$H_{\mathrm{file}}$ from the commitment, as long as the commitment
path is retained alongside the anchor. The transition from the current
design to the commitment scheme is therefore a client-side change with
no on-chain protocol modification required.

\subsection{Outlook}
\label{sec:limitations:outlook}

The following concepts are not on the near-term or medium-term
development path. We describe them because we expect questions about
long-term persistence, multi-party protocols, and substrate
independence to arise as the system matures, and because articulating
a coherent design now ensures that the system can be extended in those
directions without architectural surprises. Cross-chain migration
(3a--3c) is positioned as a long-term safety net that becomes
relevant only when the lifetime of the current substrate comes into
question. The Bilateral snaproot Commitment Protocol (BSCP) and the
AI-augmented provenance direction are medium-to-long-term items whose
realization depends on application-layer workflow development beyond
the core anchoring mechanism.

\paragraph{Cross-chain trust migration: anchoring older chain
generations within newer ones (3a).} As blockchains evolve, the
particular chain on which snaproot operates today may be superseded by
a newer protocol version. A provider --- which may be snaproot itself,
or any other party with sufficient archival capacity --- can produce a
verifiable snapshot of the relevant chain history, compute a canonical
hash over this snapshot, and anchor that hash within the newer chain
using the same snaproot mechanism. A verifier in the future can then
prove the integrity of any record on the older chain by (a) producing
the record from an archive, (b) showing that the archive snapshot
hashes to the value anchored in the newer chain, and (c) verifying the
newer-chain anchor. Diversity among independent snapshot providers
provides additional robustness: if multiple parties independently
publish snapshot hashes and they agree, manipulation is detectable.

\paragraph{Cross-chain hash migration: user-initiated, wallet-bound (3b).}
A user can elect to migrate all snaproots created under their wallet
$W$ to a different blockchain --- for example, from Solana to Ethereum.
This migration is performed as a single aggregated transaction on the
target chain. The migration record contains the source-wallet pubkey,
the migration timestamp, and a commitment to the set of migrated
records (for example, their Merkle root). The exact data structure of
the commitment is a design question that we leave open here: a pure
Merkle root minimizes on-chain footprint but requires the user to hold
the proof set externally; an explicit hash list is self-contained but
larger; a Merkle root combined with a locally retained proof set
combines the efficiency of the former with the verifiability of the
latter. We anticipate the third variant to be the most natural fit for
snaproot's user-sovereignty property, since it places the proof
material in the user's own keeping. Migration is purely user-initiated:
no party other than the holder of the source wallet's private key can
migrate records, which preserves the user-sovereignty property of the
trust architecture across the chain boundary.

\paragraph{Substrate independence of snaproot as a whole (3c).}
Mechanism 3b above is the foundation on which snaproot itself can be
ported to a different blockchain substrate in the future, while
preserving the verifiability of records created on the current
substrate. The choice of Solana today is an operational decision
optimized for the latency and cost profile of the current chain
ecosystem; it is not an architectural commitment. A future snaproot
client on a different chain would (i) anchor new records natively on
the new chain, (ii) recognize migrated records via their cross-chain
migration commitments, and (iii) optionally provide a fallback
verification path against archives of the old chain. The design space
for the user-facing presentation of pre-migration records in a
post-migration client (sorting, search, display of provenance across
chains) is part of this outlook but does not affect the underlying
cryptographic guarantees.

\paragraph{Bilateral snaproot Commitment Protocol (BSCP) for multi-party document anchoring (4).}
\label{sec:limitations:bilateral}
A recurring scenario in legal, commercial, and regulatory contexts is
that two parties wish to jointly anchor a shared document --- for
example, a contract --- such that neither can later claim a different
version was the agreed original. The naive approach of each party
independently anchoring the document is insufficient, because it
admits a manipulation attack: one party could anchor a subtly modified
version of the document \emph{before} the other party anchors the
agreed version, thereby manufacturing an earlier on-chain existence
claim for the altered content. This is a content-level frontrunning
attack, distinct from the transaction-level frontrunning discussed in
Section~\ref{sec:threat:adversary}, because it does not require
knowledge of the other party's pending transaction --- it requires only
that the attacker anchor first.

The Bilateral snaproot Commitment Protocol (BSCP) addresses this by requiring both
parties to produce on-chain evidence of agreement to the \emph{same}
hash. We describe two variants of increasing strength.

In the \emph{loosely coupled} variant, both parties independently
anchor the same hash $H = \mathrm{SHA256}(D)$ of the agreed document
$D$ under their respective wallets, within a defined time window. A
verifier confirms that (i)~both anchors carry the same hash, and
(ii)~both block times fall within the agreed window. This provides mutual
attestation without sequential dependency: neither party needs to wait
for the other's transaction to be confirmed before submitting their
own. The limitation is that it does not prove that each party had seen
the other's anchor at the time of submission; a party could in
principle anchor a different version of the document under a different
transaction and later claim their intended anchor was the loosely
coupled one.

In the \emph{tightly coupled} variant, the second party's anchor
payload includes the transaction signature of the first party's
confirmed anchor. This creates a cryptographic chain: the second
anchor provably post-dates the first and explicitly acknowledges it.
A verifier can confirm that (i) Party~A anchored $H$ at time $t_A$,
(ii) Party~B anchored $H$ at time $t_B > t_A$, and (iii) Party~B's
anchor payload references Party~A's transaction signature, proving
that Party~B was aware of Party~A's anchor. The tight coupling
eliminates the ambiguity of the loosely coupled variant but requires
sequential submission: Party~B must wait for Party~A's transaction to
be confirmed before constructing their own anchor payload.

Both variants rely on snaproot's existing wallet infrastructure: each
party submits their anchor using their own snaproot user wallet, signed
locally, with the relay wallet as fee payer. No protocol extension is
required beyond the ability to include a reference transaction
signature in the memo payload, which the current payload structure
already accommodates. The application-level workflow --- including
the time window agreement, the handshake between wallets, and the
presentation of the bilateral proof --- is the subject of ongoing
design work. The BSCP is positioned as a medium-term outlook item:
the cryptographic primitives are already present in the deployed
system, but the application-level coordination protocol requires
dedicated design and implementation effort that goes beyond the
near-term roadmap. A key open question is the handling of the
non-participation case: if one party anchors and the other does not,
the single-party anchor remains valid as a unilateral existence proof
but does not constitute a BSCP commitment; the protocol should
make this distinction explicit to users.

\paragraph{Identity anchoring pattern as a user workflow.} A user
concerned about long-term provability of authorship --- not only
existence --- can anchor an identity-attesting image as one of the
first records under their wallet (for example, a photograph of
themselves holding an identity document). Should the user later lose
their private key, the existence-proof persistence (G2) ensures that
the identity-attesting record itself remains verifiable; the user can
then re-establish the authorship link by physical means (in-person
identity verification, or automated face-against-document
comparison~\cite{spreeuwers2012biosig}) combined with the
chain-resident record. This pattern does not require any extension of
the snaproot system; it relies entirely on existing functionality and
on the structural separation between existence and authorship proofs
discussed in Section~\ref{sec:security}. It is documented here as a
recommended user practice, not as a planned feature.

\paragraph{AI-augmented integrity and provenance verification (5).}
\label{sec:limitations:ai}
The integrity guarantees that snaproot provides are deterministic and
binary: a file either matches its registered hash or it does not. This
is complementary to, but structurally distinct from, the probabilistic
judgements that AI-based content analysis produces --- for example,
whether an image appears to have been synthetically generated or
manipulated. A longer-term direction for the snaproot ecosystem is to
use the chain-resident existence record as a trusted anchor point for
downstream AI-based provenance workflows.

The envisioned integration works as follows. A file that has been
registered via capture-time anchoring carries an immutable on-chain
record of its state at the moment of capture. Any AI model subsequently
applied to that file --- for fake detection, deepfake identification,
or content authenticity classification --- can be given the guarantee
that the file it is analysing is cryptographically identical to the
registered original: the hash comparison of Algorithm~\ref{alg:verify}
establishes this before the AI model is invoked. This eliminates a
class of attacks in which an adversary substitutes a manipulated
version of the file between registration and AI analysis, causing the
model to render a verdict on the wrong input.

Conversely, the AI model's output can itself be anchored as a
snaproot record, binding the model's verdict to the specific file
version and to the specific model version that produced it. This
creates a verifiable provenance chain: the original file, its
integrity proof, and the AI-generated analysis are all
chain-resident and mutually linked, allowing any later party to
reconstruct the full provenance of the authenticity claim.

This direction is structured in two layers with different time
horizons. The first layer --- using snaproot's chain-resident existence
record as a trusted input anchor for AI-based fake detection --- is a
medium-term direction: it does not require changes to the snaproot
core and depends primarily on the availability of suitable AI models
and integration interfaces. The second layer --- anchoring the AI
model's verdict itself as a snaproot record to create a full
provenance chain --- is a longer-term direction that depends on the
maturation of standardised AI provenance interfaces, an active area
of development in initiatives such as C2PA~\cite{c2pa2023spec}. In both
cases, the specific AI models and integration points are outside the
scope of the snaproot core architecture. The snaproot contribution in this context is the
provision of a tamper-evident, provider-independent anchor that any
conforming AI provenance workflow can reference without trusting the
snaproot operator.

\section{Discussion}
\label{sec:discussion}

The technical contributions described in Sections~\ref{sec:architecture}
through~\ref{sec:limitations} support a higher-level claim that we made
in the introduction: snaproot realizes a trust architecture whose
integrity guarantees do not depend on any single trusted party and
remain verifiable across organizational, technological, and generational
boundaries. We now revisit this claim in the light of the technical
material and discuss how it manifests in two illustrative application
domains: media authenticity and compliance.

\subsection{The Four-Tier Trust Architecture, in Retrospect}

The four tiers introduced in Section~\ref{sec:introduction} are now
backed in different ways. The first three are realized by specific
technical mechanisms in the deployed system; the fourth is set out as
a coherent design rather than as a delivered mechanism. User
sovereignty is realized through locally generated keypairs, locally
performed hashing, and a verification path that requires no snaproot
infrastructure (Section~\ref{sec:security}, G5). Within-chain integrity
(Section~\ref{sec:security}, G1) reduces formally to the second-preimage
resistance of SHA-256, and to the safety of Solana consensus for the
persistence of the recorded hash. Cross-fork verifiability
(Section~\ref{sec:security}, G6) follows from the distinction between
safety and liveness: a network event that affects liveness does not
retroactively change confirmed records. Cross-chain trust migration
(Section~\ref{sec:limitations:outlook}) is the prospective fourth tier;
if realized, it would extend the architecture so that records can
outlive the specific blockchain on which they were created. For the
present system, the long-term existence guarantee is carried by the
first three tiers, with the fourth held in reserve as an additional
safety net for horizons that exceed the lifetime of any single
substrate.

A structural property that runs through all four tiers is the
\emph{separation between existence proof and authorship proof}. The
claim that a file with a given hash existed at a given time is a
chain-resident fact, verifiable by anyone with the file and any working
RPC endpoint; this claim is robust to private-key loss, organizational
change, and chain migration. The claim that a specific natural or
legal person produced the file is bound to a private key and can be
forfeited or transferred. The two claims compose in the natural way:
an identity-anchoring user workflow rebinds a lost authorship claim
to a chain-resident identity-attesting record, recovering the property
in practice while preserving the structural separation in theory.

A practical consequence worth noting explicitly: snaproot users do not
need to trust the snaproot operator to verify their own records. If
the snaproot organization were to cease operations tomorrow, every
existing snaproot record on Solana Mainnet would remain fully
verifiable using the public verification protocol of
Algorithm~\ref{alg:verify} against any Solana RPC endpoint. This is
not a contingency or a fallback mechanism; it is the default state of
the system. Trust in the snaproot operator is needed only for
convenience features --- the web frontend, the certificate rendering,
the mobile-app user interface --- none of which sit on the critical
path for cryptographic verification.

\subsection{Application Domain: Media Authenticity}
\label{sec:discussion:media}

Generative AI has made it possible to produce highly realistic
synthetic images, video, and audio that cannot be distinguished from
authentic material by visual inspection
alone~\cite{rossler2019faceforensics, li2018ictuoculi}. Approaches
based on classifying suspected content as ``authentic'' or
``synthetic'' are inherently probabilistic and degrade as generative
quality improves: they answer a question (whether a piece of content
is generated) that depends on the state of generative technology,
when the answerable question is whether a piece of content has
changed since a known reference point.

snaproot answers the latter question deterministically. Given a
published image $I$ that has been registered on-chain at time $t$,
any later version $I'$ that differs from $I$ in even a single bit
will produce a different hash and will be detected as modified.
Manipulation of an already-published image, whether by re-encoding,
deepfake overlay, or fully synthetic re-rendering of a similar scene,
leaves a chain-evident trace: the modified image's hash does not
appear in the chain, while the original's hash does.

Two sub-cases must be distinguished, as we noted in
Section~\ref{sec:architecture:capture}.

\paragraph{Post-publication manipulation.} The image is registered
before publication. After publication, an adversary creates a modified
or fully fabricated variant. snaproot detects this deterministically:
the original retains a chain-anchored existence proof; the variant
does not. The trust assumption is that the original was published
truthfully.

\paragraph{Pre-publication manipulation.} The image is fabricated by
the publisher and registered as if it were authentic. In this case
snaproot alone cannot detect the fabrication, since the registered
hash is the hash of the fabricated content. Capture-time anchoring
within a trusted capture application materially reduces this gap: the
adversary must compromise either the device or the application to
produce a fraudulent registration, rather than simply post-process a
file before registering it. Combined with hardware-attested capture
(an active area in the C2PA ecosystem~\cite{c2pa2023spec}), this
reduces the trust surface to the integrity of well-defined hardware
and software components rather than to the credibility of the
publisher.

In both sub-cases, snaproot complements rather than competes with
AI-based detection: where the latter produces a probabilistic estimate
that may be revisited as new generators emerge, snaproot produces a
binary fact that does not depend on generator-specific knowledge.

\subsection{Application Domain: Tamper-Evident Records in Regulated Environments}
\label{sec:discussion:compliance}

Many regulatory regimes require organizations to maintain integrity
of audit logs, financial records, or operational documentation for
extended retention periods. Such requirements appear in regimes such
as the German GoBD principles for digital business records,
ISO/IEC~27001 controls on logging and monitoring, GDPR data-processing
logs, SOX-related financial controls, and NIS2 incident-record
obligations. We do not claim that snaproot in itself satisfies any
specific control of these regimes; rather, we discuss it as a
technical building block that may complement an organization's
existing compliance mechanisms. The common requirement that motivates
the discussion is that records must be demonstrably unaltered after
creation, and that this demonstration must be acceptable to external
auditors.

Existing approaches to this requirement typically rely on
organization-controlled storage with limited write privileges
(WORM-style storage), on signatures using organization-held keys, or
on logging to external third-party services. Each of these places at
least some trust in a party that may have an interest in the records'
appearance. WORM-storage relies on the organization to configure and
operate it correctly; organization-signed logs can be reissued under
new keys if the originals are lost or compromised; external logging
services are subject to the same trust assumptions as any third-party
service provider.

snaproot offers a structurally different building block that may
complement these existing approaches. By anchoring the hashes of
integrity-relevant records on a public blockchain, the organization
commits the records to a substrate that the organization itself does
not control. An auditor can subsequently verify integrity of any
specific record by hashing it and looking up the corresponding
on-chain entry --- a procedure that requires no cooperation from the
organization being audited and no trust in any party that the
organization could plausibly influence. The trust surface for the
integrity check shifts from the organization (which has an incentive
in the audit outcome) to the blockchain substrate (which does not).
This is a supplementary integrity signal alongside the organization's
primary controls, not a replacement for them.

Two operational granularities are relevant for compliance use cases.
For records that are critical individually --- privilege escalations,
contract signatures, key custody operations --- each event can be
anchored separately as its own memo transaction, at a per-event cost
dominated by the Solana transaction fee
(Section~\ref{sec:evaluation}). For high-volume records such
as access logs and system audit trails, aggregated registration via
Merkle trees (a roadmap feature,
Section~\ref{sec:limitations:roadmap}) allows a single on-chain
anchoring transaction to commit to all log events within a fixed
period, amortizing the cost while preserving the per-entry
verifiability.

For both granularities, the user-sovereignty property of the trust
architecture has a specific compliance consequence: the organization
anchors records under its own wallet, and any auditor can verify those
records without going through snaproot. This avoids creating a new
third-party dependency in the compliance chain.

One caveat applies specifically in this domain. The immutability that
makes the ledger attractive for tamper-evidence is in tension with
data-protection requirements that provide for a right to erasure, such
as the GDPR. Anchoring the hash of a record raises no concern, since
the hash is not reversible to its content; but compliance deployments
must ensure that no personal data is written into cleartext metadata,
since on-chain data cannot later be erased. The hash-bound metadata
mode on the roadmap (Section~\ref{sec:limitations:roadmap}) is the
appropriate choice where record metadata must be retained in a manner
compatible with such regimes.

We close this discussion with a deliberate restraint. None of the
regulatory regimes named above recognizes blockchain anchoring as a
substitute for its primary control mechanisms; deployments must
continue to satisfy those mechanisms in their own right. The
contribution of snaproot in this context is narrower: it provides an
externally verifiable integrity signal that an auditor can check
without trusting the audited organization or any single third-party
service. A serious treatment of how this integrates with the control
catalogues of a specific regime is beyond the scope of the present
paper.

\subsection{Positioning Relative to Prior Hash-Anchoring Systems}

OpenTimestamps, OriginStamp, and Chainpoint share with snaproot the
hash-anchoring paradigm and the goal of providing integrity guarantees
that do not depend on a trusted authority. We have argued throughout
this paper that snaproot differs from these systems along three
properties: user sovereignty (records are signed by the user, not by
an operator's aggregator), latency profile (sub-five-second
confirmation on Devnet, supporting capture-time workflows under
typical conditions), and a designed-in path
toward substrate independence (cross-chain trust migration as an
explicit outlook rather than an afterthought).

The first of these is a delivered property of the trust architecture.
The second is specific to the substrate choice and is the main
technical reason that the present work positions Solana, despite its
more concentrated validator set compared to Bitcoin, as a reasonable
substrate for this application class. The third is prospective rather
than realized; its value at this stage is that the architecture leaves
room for it, not that the system already provides it. The trade-off in
the substrate choice is explicit: snaproot accepts a weaker
liveness-decentralization profile in exchange for a substantially
better latency profile, while preserving safety guarantees for
already-confirmed records and holding the cross-chain migration
outlook in reserve as a longer-term safety net.

\section{Conclusion}
\label{sec:conclusion}

We have presented snaproot, a system for decentralized file integrity
verification through blockchain-anchored cryptographic hashing on
Solana. Beyond the basic anchoring mechanism, the contribution of this
work is the articulation and partial realization of a four-tier trust
architecture that combines user sovereignty, within-chain integrity,
cross-fork verifiability, and a substrate-independent migration
outlook. We have identified the structural separation between existence
proof and authorship proof as a property that distinguishes snaproot
from designs in which a single secret carries both claims, and have
shown that this separation makes the existence guarantee robust to
the loss of any specific key and to the lifetime of any specific
operator.

We have formalized the threat model, derived the integrity guarantee
from the second-preimage resistance of SHA-256, and discussed
authorship, public verifiability, provider independence, and
cross-fork verifiability as separate security goals. We have evaluated
the system empirically on Solana Devnet, characterizing hashing
performance across file sizes from 1\,KB to 500\,MB and registration
and verification latencies under typical conditions, and have shown
that the resulting latency profile makes capture-time anchoring of
photographs a practical workflow rather than an aspirational feature.
The system is deployed on Solana Mainnet and is in productive use.

We have also been explicit about what snaproot does \emph{not}
provide: it does not certify the authenticity of pre-registration
content, it does not currently defend against frontrunning of
unconfirmed registrations, and it does not eliminate the long-term
reliance on blockchain archive availability. We have set out a roadmap
and a longer-term outlook to address these limitations as the system
matures, while keeping the core architecture stable.

snaproot is not a replacement for content authentication systems or
for compliance frameworks; rather, it provides a tamper-evident record
of a file's state at a given point in time, with a trust profile that
existing centralized integrity services cannot match. Used together
with capture-time provenance tools and with established compliance
frameworks, it can serve as a structural building block for digital
trust in legal documentation, content authentication, and
regulatory-grade record-keeping.

\section*{Author Contributions and Intellectual Property}

snaproot is an idea, design, and development of Chain Horizon GmbH.
The conception, system architecture, and implementation of snaproot
originated within Chain Horizon GmbH prior to the preparation of this
manuscript. ``snaproot'' is a registered word mark (Wortmarke) of
Chain Horizon GmbH with the German Patent and Trademark Office
(Deutsches Patent- und Markenamt, DPMA).

Arslan Br\"omme and Tarkan Yavas are affiliated with Chain Horizon
GmbH and led the conception, design, and development of the snaproot
system. They are the sole authors of this manuscript.

All rights, title, and interest in and to snaproot --- including,
without limitation, the system, its source code, the name
``snaproot'', the associated design, and the corresponding DPMA
registration --- remain exclusively with Chain Horizon GmbH.
Authorship of this manuscript does not constitute, transfer, grant,
or imply any ownership, license, or other proprietary rights in or to
snaproot or its underlying intellectual property.

\section*{Acknowledgments}

The authors gratefully acknowledge the contribution of Ramazan
Karata\c{s}, who, during his internship at Chain Horizon GmbH,
conducted the empirical evaluation work reported in
Section~\ref{sec:evaluation}, namely the measurements of hashing
performance, transaction latency, throughput, and cost.

\section*{Use of AI Tools}

The authors used the following AI-assisted tools during the preparation
of this manuscript: ChatGPT (OpenAI, accessed 2025--2026) and Claude Opus 4
(Anthropic, accessed via Claude.ai, 2025--2026). These tools were employed to
support literature research, language editing, structuring of the
manuscript, and refinement of academic writing.

All intellectual contributions reflected in this manuscript ---
including research design, analysis, interpretation, and conclusions
--- remain solely those of the human authors. Rights in the snaproot
system and its underlying intellectual property are governed by the
Author Contributions and Intellectual Property section above and
remain exclusively with Chain Horizon GmbH. AI-generated outputs were
reviewed, edited, and validated by the authors prior to inclusion.

\bibliographystyle{unsrtnat}
\bibliography{references}

\end{document}